\documentclass{article}

\usepackage{arxiv}

\usepackage[utf8]{inputenc} 
\usepackage[T1]{fontenc}    
\usepackage{hyperref}       
\usepackage{url}            
\usepackage{booktabs}       
\usepackage{amsfonts}       
\usepackage{nicefrac}       
\usepackage{microtype}      
\usepackage{lipsum}		
\usepackage{graphicx}
\usepackage{natbib}
\usepackage{doi}

\usepackage{graphicx}
\usepackage{amsmath}
\usepackage{amsthm}
\usepackage{booktabs}
\usepackage{algorithm}
\usepackage{algorithmicx}
\usepackage{algpseudocode}
\usepackage{subfigure}
\usepackage{threeparttable}
\usepackage{rotating}
\usepackage{bm}

\def\eps{{\epsilon}}
\def\mI{{\bm{I}}}
\def\mW{{\bm{W}}}
\def\gC{{\mathcal{C}}}
\def\gD{{\mathcal{D}}}
\def\gO{{\mathcal{O}}}
\def\gX{{\mathcal{X}}}
\def\gY{{\mathcal{Y}}}
\def\vzero{{\bm{0}}}
\def\vtheta{{\bm{\theta}}}
\def\vx{{\bm{x}}}
\def\vz{{\bm{z}}}
\def\sR{{\mathbb{R}}}
\def\rvb{{\mathbf{b}}}
\def\ervb{{\textnormal{b}}}

\newtheorem{definition}{Definition}
\newtheorem{thm}{Theorem}
\newtheorem{remark}{Remark}

\title{Towards Practical Differential Privacy in Data Analysis: Understanding the Effect of Epsilon on Utility in Private ERM}

\date{} 					

\author{ Yuzhe Li \\
	Institute of Information Engineering\\
	Chinese Academy of Sciences\\
	\texttt{liyuzhe@iie.ac.cn} \\
	\And
	Yong Liu\thanks{Corresponding author.} \\
	Gaoling School of Artificial Intelligence\\
	Renmin University of China, Beijing, China\\
	\texttt{liuyonggsai@ruc.edu.cn} \\
	\And
	Bo Li \\
	Institute of Information Engineering\\
	Chinese Academy of Sciences\\
	\texttt{libo@iie.ac.cn} \\
	\And
	Weiping Wang \\
	Institute of Information Engineering\\
	Chinese Academy of Sciences\\
	\texttt{wangweiping@iie.ac.cn} \\
	\And
	Nan Liu \\
	Institute of Information Engineering\\
	Chinese Academy of Sciences\\
	\texttt{liunan@iie.ac.cn} \\
}

\renewcommand{\shorttitle}{\textit{arXiv} Template}

\hypersetup{
pdftitle={A template for the arxiv style},
pdfauthor={Yuzhe Li},
}

\begin{document}
\maketitle

\begin{abstract}
As computation over sensitive data has been an important goal in recent years, privacy-preserving data analysis has gradually attracted more and more attention. Among various mechanisms, differential privacy has been widely studied due to its formal privacy guarantees for data analysis. As one of the most important issues, the crucial trade-off between the strength of privacy guarantee and the effect of analysis accuracy is highly concerned among the researchers. Existing theories for this issue consider that the analyst should first chooses a privacy requirement and then attempts to maximize the utility. However, as differential privacy is gradually deployed in practice, a gap between theory and practice comes out: in practice, product requirements often impose hard accuracy constraints, and privacy (while desirable) may not be the over-riding concern. Thus, it is usually that the requirement of privacy guarantee is adjusted according to the utility expectation, not the other way around. This gap raises the question of how to provide maximum privacy guarantee for data analysis due to a given accuracy requirement. In this paper, we focus our attention on private Empirical Risk Minimization (ERM), which is one of the most commonly used data analysis method. We take the first step towards solving the above problem by theoretically exploring the effect of $\eps$ (the parameter of differential privacy that determines the strength of privacy guarantee) on utility of the learning model. We trace the change of utility with modification of $\eps$ and reveal an established relationship between $\eps$ and utility. We then formalize this relationship and propose a practical approach for estimating the utility under an arbitrary value of $\eps$. Both theoretical analysis and experimental results demonstrate high estimation accuracy and broad applicability of our approach in practical applications. As providing algorithms with strong utility guarantees that also give privacy when possible becomes more and more accepted, our approach would have high practical value and may be likely to be adopted by companies and organizations that would like to preserve privacy but are unwilling to compromise on utility.
\end{abstract}


\section{Introduction}
In recent years, more and more researches have exposed potential privacy risks in data analysis tasks~\cite{fredrikson2015model,tramer2016stealing,shokri2017membership,ganju2018property,yeom2018privacy,Luca2019Exploiting}.  Therefore, with more and more sensitive data is involved in data analysis tasks, the privacy concerns are becoming more and more serious. To solve this problem, various mechanisms have been proposed to preserve privacy leakage during data analysis, e.g., \emph{k}-anonymity and its variants~\cite{sweeney2002k,hay2008resisting,liu2008towards,li2012on}. However, due to the limitations of the schemes (e.g., they are syntactic properties based) and the rich amount of auxiliary information available to adversaries, these mechanisms are susceptible to emerging \emph{structure based de-anonymization attacks}~\cite{backstrom2011wherefore,narayanan2009de,srivatsa2012deanonymizing}. One definition that overcome the limitations of existing mechanisms by protecting against any adversaries with any potential auxiliary information is differential privacy.

Differential privacy~\cite{dwork2006calibrating} provides a generic solution for privacy-preserving data analysis. It proves formal guarantees, in terms of a privacy budget $\eps$, on how much information is leaked. However, since approaches for achieving differential privacy necessarily lead to a decreased accuracy, it is necessary for differentially private data analysis approaches to simultaneously consider privacy guarantees and utility guarantees. Privacy alone can be easily achieved by injecting adequate noise, but this usually leads to useless analysis results. Mathematically, the trade-off between privacy and utility is tunable by $\eps$, that is, smaller $\eps$ leads to stronger privacy guarantee with lower utility, and larger $\eps$ leads to weaker privacy guarantee with higher utility. Thus, the choice of $\eps$ really matters for the implementation of differential privacy in data mining.

Traditional approaches to differential privacy assume that the user can pick the privacy parameter $\eps$ via some exogenous process, and attempt to maximize the utility of the computation subject to the privacy constraint~\cite{chaudhuri2011differentially,kifer2012private,bassily2014private,wang2017differentially,iyengar2019towards}. 
But these privacy-first approaches often encounter problems in practical applications. The first problem would be that they place an inordinate burden on the data analyst to understand differential privacy. These approaches imply a precondition that the user is capable to perceive and judge the effects of a privacy parameter. However, as the definition of differential privacy offers little insight into how this should be done, it may be beyond the capabilities of users. The second problem is that these approaches do not provide any guarantees to the users on the quality they really care about, namely the utility. Due to the lack of available theories, it is difficult for users to know in advance what impact the choice of $\eps$ will bring to the utility. 
The only theoretically sound method, namely ``utility theorem''~\cite{chaudhuri2011differentially}, often fails to provide meaningful guidance as they tend to be worst-case bounds and would be too conservative as a practical reference~\cite{Ligett2017accuracy}.
The inability to assess the impact of utility increases the uncertainty of introducing privacy guarantee, thus may introduce unpredictable risks and costs. For example, low $\eps$ would lead to low accuracy, and low accuracy may lead to frequent errors in face recognition system, which leads to the introduction of human assistance.


To solve the problem of traditional methods,~\cite{Ligett2017accuracy} considers a different perspective in the domain of private empirical risk minimization (ERM), where the goal is to give the strongest privacy guarantee possible, subject to a constraint on the acceptable accuracy. They propose a framework for developing differentially private algorithms under accuracy constraints, which allows one to choose a given level of accuracy first, and then finding the private algorithm
meeting this accuracy. However, an obvious drawback of this method is that it requires so many attempts of $\eps$ (noise subtraction), that would take a lot of computation and training time. This is often difficult to be accepted in practice, especially in large-scale learning tasks.

In this paper, we continue the exploration of solving traditional problems from accuracy-first perspective in private ERM. We provide a novel practical method for assessing the impact of $\eps$ on utility of differentially private learning algorithms due to the inherent formulaic relations between them, which would be both as accurate as empirical approach and as efficient as theoretical methods. 
The main contributions of this paper can be summarized as follows:
\begin{itemize}
\item We comprehensively analyze the effect of $\eps$ on utility in objective-perturbation private ERM algorithms. We are the first to find out that there exists a data-dependent relationship between the value of $\eps$ and the utility of final trained learning model.
\item We are the first to formalize this relationship and propose a practical approximate approach for utility analysis. We also show that our approximate approach could be easily used in practice applications.
\item We conduct a series of evaluations over three real-world datasets. Experimental results show that the estimation results of our approach is considerable close to the actual measured results, which demonstrates high estimation accuracy and broad applicability of our approach in practical applications.
\end{itemize}

\section{Preliminary}
In the rest of this paper, we use $\|\cdot\|$ to denote the $L_2$ norm. All vectors will typically be written in boldface. We will use $\mI$ to denote the identity matrix. Before delving into the details of our approach, we will first recall some basic concepts and expressions of differential privacy and differentially private learning. 
\subsection{Differential Privacy}
Differential privacy is a rigorous mathematical framework for privacy guarantee. It has recently received a significant amount of research attention for its robustness to known attacks. The typical definition of differential privacy is given in~\cite{dwork2006calibrating}, as follows:

\begin{definition}[$(\eps,\delta)$-differential privacy]\label{def-dp}
Given a randomized function $\mathcal{M} : \mathcal{D} \to \mathcal{R}$ with domain $\mathcal{D}$ and range $\mathcal{R}$, we say $\mathcal{M}$ is $(\eps,\delta)$-differential privacy if for all pairs of neighboring inputs $D,D'$ differing by one record, and for any subset of outputs $S \subseteq \mathcal{R}$, we have:  
	\[\Pr[\mathcal{M}(D) \in S]\leq e^{\eps}\times \Pr[\mathcal{M}(D') \in S]+\delta.\]
where $\eps$ is the privacy budget and $\delta$ is the failure probability.
\label{df-dp}
\end{definition}

When $\delta = 0$ we achieve a strictly stronger notion of $\eps$-differential privacy. 
A common way to achieve differential privacy is to add some randomized noise to the output, where the noise is proportional to sensitivity of function  $\mathcal{M}$:

\begin{definition}[Sensitivity]
The sensitivity of function $\mathcal{M}$ is the maximum change in the output of $\mathcal{M}$ for all possible pairs of neighboring inputs $D,D'$:
\[\Delta = \max_{D, D'} \|\mathcal{M}(D) - \mathcal{M}(D')\|.\]
\end{definition}

\subsection{Differentially Private ERM}
\begin{figure*}
\centering
\includegraphics[scale=0.7]{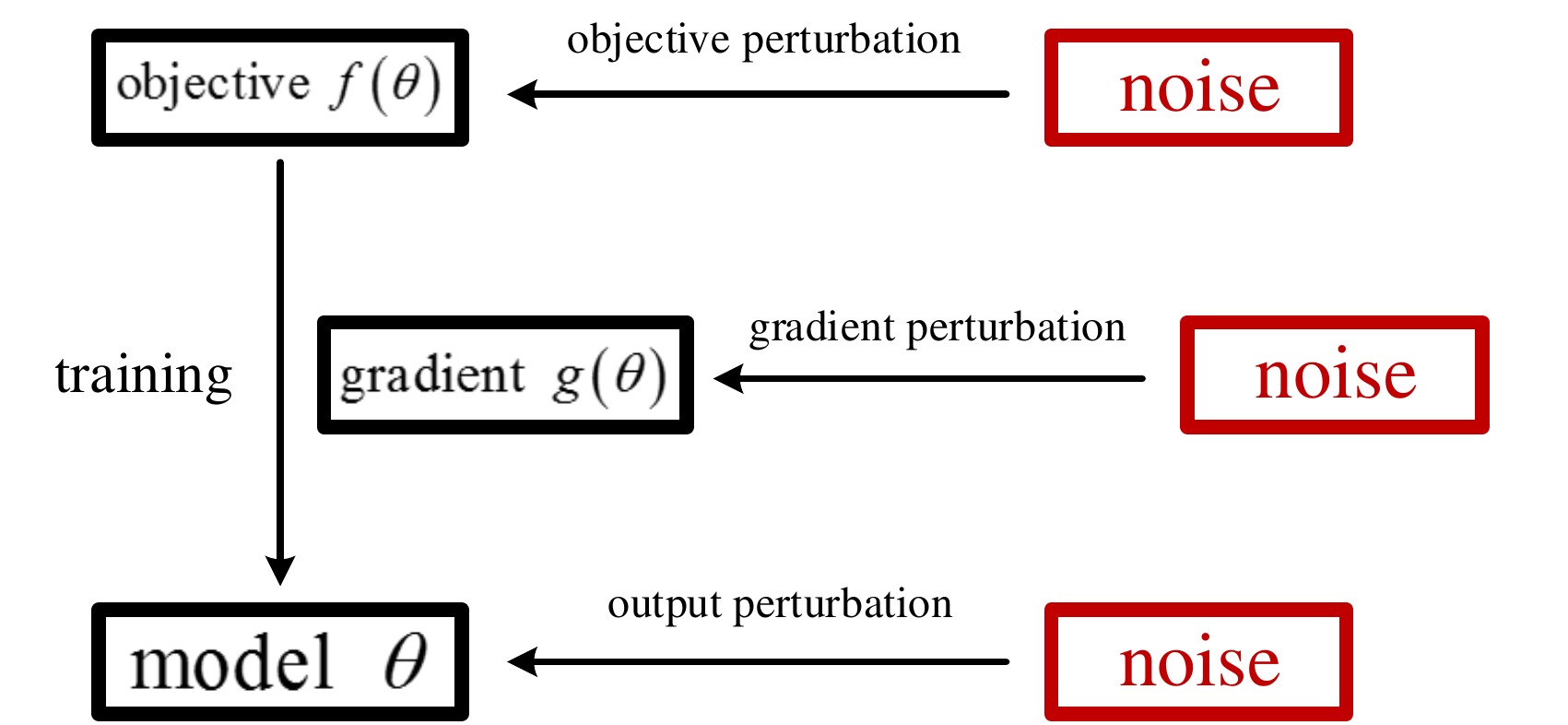}
\caption{Privacy noise mechanisms for differentially private machine learning algorithm.}
\label{fig-privacy}
\end{figure*}

In this paper, we consider the regularized ERM. Formally, let $\gX$ denote the input space and $\gY$ denote the output space, given a dataset $\gD = \{\vz_1, \vz_2,\dots,\vz_n\}$ where $\vz_i = (\vx_i, y_i) \in \gX \times \gY$, the objective of regularized ERM is to get a model $\hat{\vtheta}$ from the following unconstrained optimization problem:
\begin{equation}
\label{eq-loss}
\hat{\vtheta} = \arg\min\limits_{\vtheta}L(\vtheta, \gD) + \frac{\Lambda}{2n}||\vtheta||^2,
\end{equation}
where $L(\vtheta, \gD) = \frac{1}{n}\sum_{i=1}^n \ell(\vtheta, \vz_i)$, $\ell(\vtheta, \vz_i)$ is the loss function for $\vz_i$. The term of $\frac{\Lambda}{2n}||\vtheta||^2$ is the regularizer that prevents over-fitting. 

There are 3 obvious candidates for where to add privacy-preserving noise for machine learning algorithms (see Figure~\ref{fig-privacy}). First, we could add noise to the objective function, which gives us the objective perturbation mechanism. Second, we could add noise to the gradients at each iteration, which gives us the gradient perturbation mechanism. Finally, we can add noise to the model after the training process, which gives us the output perturbation mechanism. In this paper, we focus our research on objective perturbation mechanism among the three.

The objective perturbation mechanism was first proposed in~\cite{chaudhuri2011differentially} to make ERM algorithms achieving $\eps$-differential privacy. Later, the range of problems to which objective perturbation applies has been extended in~\cite{kifer2012private}, including convex ERM problems, non-differentiable regularizers and tighter error bound. 

\begin{thm}[Private Convex Optimization via Objective Perturbation in~\cite{kifer2012private}]
\label{objective}
Let $\gC$ be a closed convex subset of $\sR^p$, $\gD = \{\vz_1, \vz_2,\dots,\vz_n\}$ be a dataset, 
$\ell(\vtheta, \vz_i)$ be convex loss function with continuous Hessian, $||\nabla \ell(\vtheta, \vz_i)|| \le \zeta$ and $||\nabla^2 \ell(\vtheta, \vz_i)|| \le \lambda$ (for all $\vz$ and for all $\vtheta \in \gC$), $L(\vtheta, \gD) = \frac{1}{n}\sum_{i=1}^n \ell(\vtheta, \vz_i)$, $\Delta_{\eps} = \frac{2\lambda}{\eps}$. 
Then 
\begin{equation}
\label{eq-bg1}
\hat{\vtheta}_{\eps} = \arg\min\limits_{\vtheta} L(\vtheta, \gD) + \frac{\Lambda}{2n}||\vtheta||^2 + \frac{1}{n}\rvb_{\eps}^T \vtheta + \frac{\Delta_{\eps}}{2n} ||\vtheta||^2
\end{equation}
is $(\eps,\delta)$-differentially private if  $\rvb_{\eps} \sim 
\mathcal{N}\left(0, \frac{\zeta^2(8\log \frac{2}{\delta}+4\eps)}{\eps^2}\mI\right)$.
\end{thm}

\begin{remark}
The reasons why we choose objective perturbation (not output perturbation or gradient perturbation) as the basis of our analysis are as follows: a) as one of the most important mechanisms for differentially private algorithms, objective perturbation has shown superior performance on general objectives~\cite{chaudhuri2011differentially,kifer2012private,iyengar2019towards}; b) for objective perturbation, our analysis of utility and $\eps$ can be directly established on the objective function, which would make the process intuitive and concise. 
\end{remark}

\section{Traditional Utility Guarantee Analysis in Private ERM}
The exploration of the effect of $\eps$ on utility in private ERM started from the very beginning in~\cite{chaudhuri2011differentially}, where the differentially private ERM algorithms were first proposed. They found that differential privacy not only provides privacy guarantee, but also reduces the model utility. They measure the utility by the number of samples $n$ required to achieve generalization error $\xi$ on the empirical loss of a reference predictor $\vtheta_0$, and bound $n$ by using the analysis method of~\cite{Shalev2008SVM}. Their bounds provide the sample requirement of differentially private algorithms for achieving generalization error. That is, the algorithms will achieve higher utility with more training data, and as long as there is abundant training data, the utility of differentially private algorithms can be very close to non-private one. They also show that the error rate increases as the privacy requirements are made more stringent, which they call the ``price of privacy''. This is the first explicit reference to the tradeoff between privacy and utility of differentially private learning algorithms.

However, this form of expression is difficult to intuitively understand how $\eps$ effects the model utility, neither the utility loss nor the empirical loss can be acquired. To this end,~\cite{kifer2012private} provided an improved, data-dependent utility analysis, which used empirical risk $J(\vtheta^{priv}, \gD) - J(\hat{\vtheta}, \gD)$ to measure the utility, where $\hat{\vtheta}$ denotes the optimal predictor, $\vtheta^{priv}$ denotes the private predictor and $J(\vtheta, \gD)$ denotes the empirical loss. By using the tail bounds for the noise distributions used in Theorem~\ref{objective}, they obtain the following empirical risk bound for their objective perturbed learning algorithm. 

\begin{thm}[Empirical Risk Bound of Differentially Private Learning Algorithm in~\cite{kifer2012private}]
\label{theorem-bound}
Let $\hat{\vtheta}$ be the true minimizer of $J(\vtheta, \gD)$ over the closed convex set $\gC$ and let $\hat{\vtheta}_{\eps}$ be the output of (\ref{eq-bg1}) in Theorem~\ref{objective}. We have $\mathbb{E}\left[ J(\hat{\vtheta}_{\eps}, \gD) - J(\hat{\vtheta}, \gD) \right] = O \left( \frac{\zeta ||\hat{\vtheta}||\sqrt{p \log(1/\delta)}}{\eps n} \right)$.
\end{thm}

Theorem~\ref{theorem-bound} clearly illustrates the upper bound of utility loss in the worst case. This method and its adaptions have been widely adopted up to the present literatures~\cite{bassily2014private,wang2017differentially,iyengar2019towards} for measuring the utility of private learning algorithms. Although these ``utility theorems'' can give an intuitive impression of the effect of $\eps$ on utility, their bounds are too conservative to find an appropriate value of $\eps$ (usually leading to a much larger value of $\eps$). This makes them difficult to do something useful in practical applications.

In order to overcome this conservatism and give a more realistic analysis of the effect of $\eps$ on utility, we propose an ``empirical'' theorem approach. 

\section{Our Approach}
In this section, we will introduce our approach for analyzing the effect of $\eps$ on the utility with gradual steps.

First of all, we formalize our final goal by asking the following question: how would the utility of the private model change if we modify the value of $\eps$? To answer the question, we will begin by studying the change in model parameters due to modifying the value of $\eps$. Formally, this change is $\hat{\vtheta}_{\bar{\eps}} - \hat{\vtheta}_{\hat{\eps}}$, where $\hat{\vtheta}_{\bar{\eps}}$ and $\hat{\vtheta}_{\hat{\eps}}$ are the true minimizers of Equation (\ref{eq-bg1}) that achieves $(\bar{\eps}, \delta)$- and $(\hat{\eps}, \delta)$-differential privacy on the same dataset $\gD$, respectively, for a certain $\delta$. Apparently, our analysis focuses on the influence of $\eps$ on the parameters $\hat{\vtheta}_{\eps}$, which is $\frac{\partial \hat{\vtheta}_{\eps}}{\partial \eps}$.

\begin{thm}
\label{thm-1}
Let $L(\vtheta, \gD), \Delta_{\eps}$ and $\hat{\vtheta}_{\eps}$ be defined as in Theorem~\ref{objective}.
Let $\mW_{\eps} = \nabla^2 L(\hat{\vtheta}_{\eps}, \gD)+\frac{\Lambda}{n}+\frac{1}{n}\Delta_{\eps}, \rvb'_{\eps} = \frac{\partial \rvb_{\eps}}{\partial \eps}, \Delta'_{\eps} = \frac{\partial \Delta_{\eps}}{\partial \eps}$. Then, we have
\begin{equation}
\label{eq-thm1}
\frac{\partial \hat{\vtheta}_{\eps}}{\partial \eps} = - \frac{1}{n} \mW_{\eps}^{-1}\left(\rvb'_{\eps} + \Delta'_{\eps} \hat{\vtheta}_{\eps}\right).
\end{equation}
\end{thm}

\begin{proof}
Due to the fact that the derivative of the objective function (\ref{eq-bg1}) at $\hat{\vtheta}_{\eps}$ is $\vzero$, we have that

\begin{equation}
\label{eq-1}
\frac{\partial  L(\hat{\vtheta}_{\eps}, \gD)}{\partial \vtheta} + \frac{1}{n} \left( \Lambda \hat{\vtheta}_{\eps} + \rvb_{\eps} + \Delta_{\eps} \hat{\vtheta}_{\eps} \right) = \vzero.
\end{equation}
Take the derivative of (\ref{eq-1}) with respect to $\eps$, we have:
\begin{equation}\nonumber
\frac{\partial^2 L(\hat{\vtheta}_{\eps}, \gD)}{\partial \vtheta^2} \frac{\partial \hat{\vtheta}_{\eps}}{\partial \eps} + \frac{1}{n} \left(  \Lambda \frac{\partial \hat{\vtheta}_{\eps}}{\partial \eps} + \rvb'_{\eps} + \frac{\partial (\Delta_{\eps}\hat{\vtheta}_{\eps})}{\partial \eps} \right) = \vzero.
\end{equation}
Re-arranging the terms, we get
\begin{equation}\nonumber
\left(\nabla^2 L(\hat{\vtheta}_{\eps}, \gD)+\frac{ \Lambda}{n}+\frac{1}{n}\Delta_{\eps} \right) \frac{\partial \hat{\vtheta}_{\eps}}{\partial \eps} = -\frac{1}{n}\left(\rvb'_{\eps} +  \Delta'_{\eps} \hat{\vtheta}_{\eps}\right).
\end{equation}
Finally, we have
\begin{equation}
\label{eq-3}
\frac{\partial \hat{\vtheta}_{\eps}}{\partial \eps} = - \frac{1}{n} \mW_{\eps}^{-1}\left(\rvb'_{\eps} + \Delta'_{\eps} \hat{\vtheta}_{\eps}\right).
\end{equation}
\end{proof}

To obtain $\frac{\partial \hat{\vtheta}_{\eps}}{\partial \eps}$, we need to derive $\rvb_{\eps}$ over $\eps$. However, since $\rvb_{\eps}$ doesn't have an explicit function expression, we can not compute $\frac{\partial \rvb_{\eps}}{\partial \eps}$ directly. Thus, we need to first form a function that outputs $\rvb_{\eps}$. To this end, we make use of the reparameterization trick, which is commonly applied to training variational autoencoders with continuous latent variables using backpropagation~\cite{Kingma2014AutoEncodingVB,Jimenez2013Stochastic}. For a given normal distribution $z \sim \mathcal{N}(\mu, \sigma^2)$, it can be re-written as $z = \mu + \sigma \cdot \mathcal{N}(0, 1)$ by reparameterization, making it trivial to compute $\frac{\partial z}{\partial \sigma}$ and $\frac{\partial z}{\partial \mu}$.

Now we will show how to use the reparameterization trick to compute $\frac{\partial \rvb_{\eps}}{\partial \eps}$. Let $\rvb_{\eps} = (\ervb_1, \dots, \ervb_d)$, then $\rvb'_{\eps} = (\ervb'_1, \dots, \ervb'_d)$. Since $\rvb_{\eps} \sim \mathcal{N}\left(0, \frac{\zeta^2(8\log \frac{2}{\delta}+4\eps)}{\eps^2}\mI\right)$, we have that for each $i \in [d]$, $\ervb_i$ is independently sampled from $\mathcal{N}\left(0, \frac{\zeta^2(8\log \frac{2}{\delta}+4\eps)}{\eps^2}\right)$. By using reparameterization trick, we obtain $\ervb_i = \left( \frac{\zeta\sqrt{8\log \frac{2}{\delta}+4\eps}}{\eps}\right) \cdot \mathcal{N}(0,1)$. Note that $\mathcal{N}(0,1)$ can be seen as a constant term for $\eps$, we have that $\ervb'_i =\left( \frac{\zeta\left( 4\log \frac{2}{\delta}+\eps \right)}{\eps^2 \sqrt{2\log \frac{2}{\delta}+\eps}} \right) \cdot \mathcal{N}(0,1)$. Thus, we have:
\begin{equation}
\label{partial-b}
\rvb'_{\eps}= \left( \frac{\zeta\left( 4\log \frac{2}{\delta}+\eps \right)}{\eps^2 \sqrt{2\log \frac{2}{\delta}+\eps}} \right) \cdot \mathcal{N}(0,\mI).
\end{equation}

The assumptions in Theorem~\ref{thm-1} can be satisfied on many commonly used loss functions, such as logistic regression, Huber SVM and quadratic loss, and regularizers, such as $L_2$-norm regularizer. However, this does exclude some classes of loss functions (e.g. hinge loss). In Section~\ref{sec-violation}, we will talk about this cases and show that some of these assumptions can be weakened in certain cases.

Once $\frac{\partial \hat{\vtheta}_{\eps}}{\partial \eps}$ is calculated, we can further deduce the effect of $\eps$ on the utility. In particular, let $F(\vtheta, \gD)$ be the utility of $\vtheta$ on dataset $\gD$, then, by applying the chain rule, we have the following expression:
\begin{equation}
\label{equation-F}
\frac{\partial F(\hat{\vtheta}_{\eps}, \gD)}{\partial \eps} = \left( \nabla_{\vtheta} F(\hat{\vtheta}_{\eps}, \gD) \right)^T \frac{\partial \hat{\vtheta}_{\eps}}{\partial \eps}
\end{equation}


Now we can answer the question posed at the beginning of this section by utilizing the Taylor expansion of $F(\bar{\vtheta}_{\bar{\eps}}, \gD) - F(\hat{\vtheta}_{\hat{\eps}}, \gD)$ at $\bar{\eps}$: 
\begin{equation}
\label{eq-utility}
F(\hat{\vtheta}_{\bar{\eps}}, \gD) - F(\hat{\vtheta}_{\hat{\eps}}, \gD) \approx \frac{\partial F(\hat{\vtheta}_{\bar{\eps}}, \gD)}{\partial \eps}(\bar{\eps}-\hat{\eps}).
\end{equation}

Intuitively, equation (\ref{eq-utility}) formulizes the relationship between the change in $\eps$ and the change in utility. Furthermore, once we have obtained $F(\hat{\vtheta}_{\bar{\eps}}, \gD)$ with a certain $\bar{\eps}$ (through one round training), (\ref{eq-utility}) represents the relationship between $\hat{\eps}$ and its utility, which could be used for choosing $\eps$ according to expected utility. This application would be further discussed in Section~\ref{sec-choose}. Generally, $\bar{\eps}$ in (\ref{eq-utility}) is called measuring point, and $\hat{\eps}$ in (\ref{eq-utility}) is called target point.

In essence, both $F(\hat{\vtheta}_{\eps}, \gD)$ in (\ref{equation-F}) and $L(\hat{\vtheta}_{\eps}, \gD)$ in (\ref{eq-thm1}) are used to measure the error between predictions and true labels, which sometimes may even be the same in  applications. Thus, for the simplicity of exposition, we assume that all the assumptions in $L(\hat{\vtheta}_{\eps}, \gD)$ are equally valid in $F(\hat{\vtheta}_{\eps}, \gD)$.

\section{Practical Considerations}
In this section, we are going to talk about two practical considerations of our approximation approach. 


\subsection{Assumption Violation}
\label{sec-violation}
Theorem~\ref{thm-1} relies on several assumptions that may be violated in practice. The first would be non-convex or non-convergent objective functions. In this case, the obtained parameters $\tilde{\vtheta}_{\eps}$ may not be the global minimum, thus Equation (\ref{eq-1}) may not hold. In such cases, we can form a convex quadratic approximation of the loss around $\tilde{\vtheta}_{\eps}$, i.e., 
\begin{equation}\nonumber
\resizebox{.99\linewidth}{!}{$
L(\vtheta, \gD) \approx L(\tilde{\vtheta}_{\eps}, \gD) + \nabla_{\vtheta}L(\tilde{\vtheta}_{\eps}, \gD)(\tilde{\vtheta}_{\eps} - \vtheta) + (\tilde{\vtheta}_{\eps} - \vtheta)^T\nabla^2_{\vtheta}L(\tilde{\vtheta}_{\eps}, \gD)(\tilde{\vtheta}_{\eps} - \vtheta),
$}
\end{equation}
which corresponds to adding $L_2$ regularization on $\vtheta$.

The second violation would be non-differentiable losses, in which neither $\nabla_{\vtheta}L(\hat{\vtheta}_{\eps}, \gD)$ nor $\nabla^2_{\vtheta}L(\hat{\vtheta}_{\eps}, \gD)$ does exist. In these cases of violation, we can train non-differentiable models and swap out non-differentiable components for smoothed versions (as suggested in~\cite{Koh2017Understanding}). For example, we can approximated $Hinge(s) = max(0, 1-s)$ with following form $SmoothHinge(s,t) = t\log(1 + e^{\frac{1-s}{t}})$, which approaches the hinge loss at $t \to 0$.
\begin{remark}
The above analysis shows that our approximate approach can be extended to the non-convex, non-convergent and non-differentiable situations. However, in this paper, we mainly focus on verifying the effectiveness of our approach. Therefore, we only consider the normal case in our evaluations. We leave the important analysis of our approximate approach for non-convex, non-convergent and non-differentiable cases to our future work.
\end{remark}
\subsection{Error Analysis}
\label{sec-error-ana}
To analyze the error bound of our approach, we need to  complete the Taylor expression of Equation (\ref{eq-utility}) as follows:
\begin{equation}
F(\hat{\vtheta}_{\bar{\eps}}, \gD) - F(\hat{\vtheta}_{\hat{\eps}}, \gD) \approx \frac{\partial F(\hat{\vtheta}_{\eps}, \gD)}{\partial \eps}(\bar{\eps}-\hat{\eps}) + r(\bar{\eps}-\hat{\eps}),
\end{equation}
where $r(\bar{\eps}-\hat{\eps}) = \frac{\partial^2 F(\hat{\vtheta}_{\tilde{\eps}}, \gD)}{\partial \eps^2}(\bar{\eps}-\hat{\eps})^2$ is the Taylor remainder, $\tilde{\eps} \in [\bar{\eps},\hat{\eps}]$. Obviously, our error mainly comes from omitting the Taylor remainder. Thus, to derive the error bound of Equation (\ref{eq-utility}), we require the Taylor remainder to be bounded. 


\begin{thm}
\label{thm-bound}
Assume that $\ell(\vtheta, \gD)$ is third-differentiable and convex in $\vtheta$ with continuous Hessian, $||\nabla \ell(\vtheta, \vz_i)|| \le c$, $||\nabla^2 \ell(\vtheta, \vz_i)|| \le \lambda$, $\Delta_{\eps}=\frac{2\lambda}{\eps}$, $\rvb_{\eps} \sim \mathcal{N}\left(0, \frac{c^2(8\log \frac{2}{\delta}+4\eps)}{\eps^2}\mI\right)$. If there exists a constant $s$ for $\forall \vz \in \gD, \vtheta \in \gC, ||\nabla^3_{\vtheta}\ell(\vtheta, \gD)|| \le s$. Then, for any $\bar{\eps},\hat{\eps}$ we have
\begin{equation}
\label{eq-bound}
|F(\hat{\vtheta}_{\bar{\eps}}, \gD) - F(\hat{\vtheta}_{\hat{\eps}}, \gD)| = \gO\left(\frac{(\bar{\eps}-\hat{\eps})^2}{n\tilde{\eps}^3}\right)
\end{equation}
where $\tilde{\eps} \in [\bar{\eps},\hat{\eps}]$, $n$ is the number of training samples. $\gO$ omitted the terms $p, \lambda, \delta, c, s$.
\end{thm}

\begin{proof}
Let's try expanding $\frac{\partial^2 F(\hat{\vtheta}_{\tilde{\eps}}, \gD)}{\partial \eps^2}$ for the first time:
\begin{equation}
\label{proof-1-1}
\resizebox{.95\linewidth}{!}{$
    \displaystyle
\begin{aligned}
\frac{\partial^2 F(\hat{\vtheta}_{\tilde{\eps}}, \gD)}{\partial \eps^2} 
&= \frac{\partial \left( \frac{\partial F(\hat{\vtheta}_{\tilde{\eps}}, \gD)}{\partial \eps} \right)}{\partial \eps} \\
&= \frac{\partial \left( \left( \frac{\partial F(\hat{\vtheta}_{\tilde{\eps}}, \gD)}{\partial \vtheta} \right) ^T \frac{\partial \hat{\vtheta}_{\tilde{\eps}}}{\partial \eps} \right)}{\partial \eps} \\
&= \left( \frac{\partial^2 F(\hat{\vtheta}_{\tilde{\eps}}, \gD)}{\partial \vtheta \partial \eps} \right)^T \frac{\partial \hat{\vtheta}_{\tilde{\eps}}}{\partial \eps} + \left( \frac{\partial F(\hat{\vtheta}_{\tilde{\eps}}, \gD)}{\partial \vtheta} \right)^T \frac{\partial^2 \hat{\vtheta}_{\tilde{\eps}}}{\partial \eps^2} \\
&= \left( \frac{\partial \hat{\vtheta}_{\tilde{\eps}}}{\partial \eps} \right) ^T \left( \nabla^2 F(\hat{\vtheta}_{\tilde{\eps}}, \gD) \right) \frac{\partial \hat{\vtheta}_{\tilde{\eps}}}{\partial \eps} + \left( \nabla F(\hat{\vtheta}_{\tilde{\eps}}, \gD) \right)^T \frac{\partial^2 \hat{\vtheta}_{\tilde{\eps}}}{\partial \eps^2}.
\end{aligned}
$}
\end{equation}

To go on, we need to expand $\frac{\partial^2 \hat{\vtheta}_{\tilde{\eps}}}{\partial \eps^2}$, which is as follows:
\begin{equation}
\label{proof-1-2}
\resizebox{.95\linewidth}{!}{$
    \displaystyle
\begin{aligned}
\frac{\partial^2 \hat{\vtheta}_{\tilde{\eps}}}{\partial \eps^2}
&= \frac{\partial \left( \frac{\partial \hat{\vtheta}_{\tilde{\eps}}}{\partial \eps} \right)}{\partial \eps} \\
&= \frac{\partial \left(- \frac{1}{n}\mW_{\tilde{\eps}}^{-1} \left( \rvb'_{\tilde{\eps}} + \Delta'_{\tilde{\eps}} \hat{\vtheta}_{\tilde{\eps}} \right) \right)}{\partial \eps} \\
&= -\frac{1}{n} \frac{\partial \left(\mW_{\tilde{\eps}}^{-1} \right)}{\partial \eps} \left(\rvb'_{\tilde{\eps}} + \Delta'_{\tilde{\eps}} \hat{\vtheta}_{\tilde{\eps}} \right) - \frac{1}{n} \mW_{\tilde{\eps}}^{-1}\frac{\partial \left(\rvb'_{\tilde{\eps}} + \Delta'_{\tilde{\eps}} \hat{\vtheta}_{\tilde{\eps}} \right)}{\partial \eps} \\
&= \mW_{\tilde{\eps}}^{-1} \frac{\partial \mW_{\tilde{\eps}}}{\partial \eps} \left( -\frac{1}{n} \right) \mW_{\tilde{\eps}}^{-1} \left( \rvb'_{\tilde{\eps}} + \Delta'_{\tilde{\eps}}\hat{\vtheta}_{\tilde{\eps}} \right) - \frac{1}{n} \mW_{\tilde{\eps}}^{-1} \left( \frac{\partial (\rvb'_{\tilde{\eps}})}{\partial \eps} + \frac{\partial (\Delta'_{\tilde{\eps}}) }{\partial \eps} \hat{\vtheta}_{\tilde{\eps}} + \Delta'_{\tilde{\eps}}\frac{\partial \hat{\vtheta}_{\tilde{\eps}}}{\partial \eps} \right) \\
&= \mW_{\tilde{\eps}}^{-1} \left( \frac{\partial \left( \nabla^2 L(\hat{\vtheta}_{\tilde{\eps}}, \gD) \right)}{\partial \eps} + \frac{1}{n} \Delta'_{\tilde{\eps}}\mI \right) \frac{\partial \hat{\vtheta}_{\tilde{\eps}}}{\partial \eps} - \frac{1}{n} \mW_{\tilde{\eps}}^{-1} \left(\rvb''_{\tilde{\eps}} + \Delta''_{\tilde{\eps}} \hat{\vtheta}_{\tilde{\eps}} + \Delta'_{\tilde{\eps}} \frac{\partial \hat{\vtheta}_{\tilde{\eps}}}{\partial \eps} \right) \\
&= \mW_{\tilde{\eps}}^{-1} \left( \nabla^3 L(\hat{\vtheta}_{\tilde{\eps}}, \gD) \right) \left( \frac{\partial \hat{\vtheta}_{\tilde{\eps}}}{\partial \eps} \right) \left( \frac{\partial \hat{\vtheta}_{\tilde{\eps}}}{\partial \eps} \right) - \frac{1}{n} \mW_{\tilde{\eps}}^{-1} \left( \rvb''_{\tilde{\eps}} + \Delta''_{\tilde{\eps}} \hat{\vtheta}_{\tilde{\eps}} \right), 
\end{aligned}
$}
\end{equation}
where $ \rvb''_{\tilde{\eps}} =  \frac{\partial (\rvb'_{\tilde{\eps}})}{\partial \eps}$, $ \Delta''_{\tilde{\eps}} =  \frac{\partial (\Delta'_{\tilde{\eps}})}{\partial \eps}$.


Now, let's bound the term $||\mW_{\tilde{\eps}}^{-1}||$ first. By the assumptions in Theorem~\ref{thm-1}, we have
\begin{equation}
\begin{aligned}
||\mW_{\tilde{\eps}}||
&=||\nabla^2 L(\vtheta, \gD)|| + \frac{\Lambda}{n} + \frac{1}{n}||\Delta_{\eps}|| \\
&=\frac{1}{n}\sum_{i=1}^n||\nabla^2 \ell(\vtheta, \vz_i)|| + \frac{\Lambda}{n} + \frac{2\lambda}{n\eps} \\
&= \gO \left( \lambda + \frac{\Lambda}{n} + \frac{2\lambda}{n\eps} \right).
\end{aligned}
\end{equation}

Thus, $||\mW_{\tilde{\eps}}^{-1}|| = \gO \left( \frac{1}{\lambda + \frac{\Lambda}{n} + \frac{2\lambda}{n\eps}} \right)$. 
Since it is usally that $\eps \gg \frac{1}{n}$, we obtain the bound of $\mW_{\tilde{\eps}}^{-1}$ to be $\gO \left( \frac{1}{\lambda} \right)$.

As $\Delta_{\eps} = \frac{2\lambda}{\eps}$, we have that $\Delta'_{\tilde{\eps}} = \gO \left( \frac{\lambda}{\eps^2} \right)$ and $\Delta''_{\tilde{\eps}} = \gO \left( \frac{\lambda}{\eps^3} \right)$.

Then we bound $\rvb'_{\tilde{\eps}}$ and $\rvb''_{\tilde{\eps}}$. 
We find $\log \frac{2}{\delta} + \eps$ and $\log \frac{2}{\delta}$ are of the same order of magnitude as large $\eps$ values are rarely used. Thus we have that $\rvb'_{\tilde{\eps}} = \gO \left( \frac{\sqrt{\log (\frac{1}{\delta})}}{\eps^2} \right)$ acoording to Equation (\ref{partial-b}). Similarly, we have that $\rvb''_{\tilde{\eps}} = \gO \left( \frac{\sqrt{\log (\frac{1}{\delta})}}{\eps^3} \right)$.

Note that $\frac{\partial \hat{\vtheta}_{\eps}}{\partial \eps}$ can be easily calculated by Equation (\ref{eq-3}), $||\vtheta||$ is usually assumed to be at most $\sqrt{p}$ via normalization ($p$ is the dimensionality of data, which is not considered in our paper), $||\nabla L(\hat{\vtheta}_{\tilde{\eps}}, \gD)|| \le c$, $||\nabla^2 L(\hat{\vtheta}_{\tilde{\eps}}, \gD)|| \le \lambda$ and $||\nabla^3 L(\hat{\vtheta}_{\tilde{\eps}}, \gD)|| \le s$, we can obtain our error bound of $r(\bar{\eps}-\hat{\eps}) = \frac{\partial^2 L(\hat{\vtheta}_{\tilde{\eps}}, \gD)}{\partial \eps^2}(\bar{\eps}-\hat{\eps})^2$ as follows:
\[\gO \left( \frac{(\bar{\eps}-\hat{\eps})^2}{n\tilde{\eps}^3} \right),\]
where $\gO$ omitted the terms $p, \lambda, \delta, c, s$.
\end{proof}

To bound the Taylor remainder, we assume the loss function has bounded third derivative. For most commonly used twice-differentiable loss functions, this assumption is easy to be satisfied. For example, the third derivative of quadratic loss is zero, any order derivative of logistic regression is bounded, etc.. 

Theorem~\ref{thm-bound} shows that our error bound is proportional to the difference between $\bar{\eps}$ and $\hat{\eps}$, and inversely proportional to sample number $n$. That is, the more samples we used for training and the smaller the change in $\eps$, the higher estimation accuracy we could obtain. Moreover, the magnitude of $\tilde{\eps}$ (i.e. the magnitude of $\bar{\eps}$ and $\hat{\eps}$) is also an important factor affecting estimation accuracy. Considering the general cases, where $\bar{\eps}$ and $\hat{\eps}$ are of the same magnitude, our bound would approach to $\gO\left(\frac{1}{n\tilde{\eps}}\right)$, which is fairly  small when $\tilde{\eps}$ is of general magnitude (i.e. $10^{-2} \sim 10^1$). However, if $\bar{\eps}$ and $\hat{\eps}$ are of different magnitudes, especially when one is of small magnitude, the error may be relatively large. For example, when $\bar{\eps} = 0.01, \hat{\eps}=10$, we have $(\bar{\eps}-\hat{\eps})^2  \approx 100$ and $\tilde{\eps} \in [0.01, 10]$, then the error would be between $\gO(\frac{10^{-1}}{n})$ and $\gO(\frac{10^8}{n})$. Thus, to obtain small error, we should ensure that $\bar{\eps}$ and $\hat{\eps}$ are of the same magnitude when using our approximation approach.

\section{Practical Use Case of Our Approach}
\label{sec-choose}

By revealing insights about how $\eps$ affects the model utility in differentially private ERM, our approach would be very useful in the practical application. In this section, we will show a typical use case about how our approach guides users to choose an appropriate $\eps$ for private learning in practice.



The question of how to choose $\eps$ is a typical application scenario, which has existed since differential privacy was first proposed~\cite{lee2011how,hsu2014differential}. Traditional approaches to differential privacy assume a fixed privacy requirement $\eps$ for the user and attempt to maximize the utility of model to the privacy constraint~\cite{naldi2015differential,hsu2014differential,8511827}. However, existing approaches share a common precondition that the user is capable to perceive and judge the effects of a privacy parameter. Considering that the privacy is hard to be quantitative variable, and the definition of differential privacy offers little insight into how this should be done, it may be a little harsh to require users to have such ability.

Different from the existing approaches,  we provide a more efficient and accurate approach to addressing this issue, based on the conclusions of the previous analysis. Conceptually, by re-arranging the terms in (\ref{eq-utility}), we have:
\begin{equation}
\label{eq-implement}
\hat{\eps} \approx \frac{F(\hat{\vtheta}_{\hat{\eps}}, \gD) - F(\hat{\vtheta}_{\bar{\eps}}, \gD)}{\frac{\partial F(\hat{\vtheta}_{\bar{\eps}}, \gD)}{\partial \eps}} + \bar{\eps}.
\end{equation}
Taking $F(\hat{\vtheta}_{\hat{\eps}}, \gD)$ as the expected utility we want to satisfy, then $\hat{\eps}$ would be the approximation of the $\eps$ we are looking for (i.e. the minimum $\eps$ achieves $F(\hat{\vtheta}_{\hat{\eps}}, \gD)$). Thus, once we obtain $F(\hat{\vtheta}_{\bar{\eps}}, \gD)$ and $\frac{\partial F(\hat{\vtheta}_{\bar{\eps}}, \gD)}{\partial \eps}$ by training a model $\hat{\vtheta}_{\bar{\eps}}$ with $\bar{\eps}$, $\hat{\eps}$ can be directly calculated without repeated trainings. Algorithm~\ref{alg-choose} gives a brief outline of this process.

\begin{algorithm}[pt]
	\caption{Choosing $\eps$ with our approximation approach}
	\begin{algorithmic}[1]
	\Require Dataset $\gD = \{\vz_1, \vz_2,\dots,\vz_n\}$, loss function $\ell(\vtheta, \vz_i)$, measuring privacy parameter $\bar{\eps}$, failure probability $\delta$, loss function $\ell(\vtheta, \vz_i)$, expected utility $\hat{F}$.
	\State Train a private model $\hat{\vtheta}_{\bar{\eps}}$ on $\gD$ by Theorem~\ref{objective}.
	\State Compute $\frac{\partial F(\hat{\vtheta}_{\bar{\eps}}, \gD)}{\partial \eps}$ by Theorem~\ref{thm-1}
	\Ensure $\hat{\eps}=\frac{\hat{F}-F(\hat{\vtheta}_{\bar{\eps}}, \gD)}{\frac{\partial F(\hat{\vtheta}_{\bar{\eps}}, \gD)}{\partial \eps}} + \bar{\eps}$
	\end{algorithmic}
\label{alg-choose}
\end{algorithm}


Our approach mainly has two advantages for choosing $\eps$. Compared with utility theorem, our approach can find nearly optimal $\eps$ that is very close to the true minimum value consistent with utility requirement as we analyzed in Section~\ref{sec-error-ana}. And, compared with empirical search, our approach is much more efficient as we only need to train the algorithm on the full dataset once. This feature will play a great advantage in case large-scale learning tasks.



\section{Evaluation}
\subsection{Evaluation Setup}
In this section, we will empirically analyze the performance of our approach. We implement the objective perturbation mechanism of~\cite{kifer2012private} based on the open source released by~\cite{iyengar2019towards}. 
Our evaluation considers the loss functions for two commonly used models: Logistic Regression (LR) and Huber SVM (SVM), which are defined as follows:
\begin{equation}
\begin{split}
\label{eq-lr}
\ell_{LR}(\vz) &= \log(1+e^{-\vz}), \\
\ell_{SVM}(\vz) &= 
\begin{cases}
0, & \mbox{if } z > 1 + h, \\
\frac{1}{4h}(1 + h - z)^2,  & \mbox{if } |1-z| \le h, \\
1-z, & \mbox{if} z<1-h,
\end{cases}
\end{split}
\end{equation}
both of which are twice-differentiable. The term $h$ in SVM is set to be $0.1$ by default. The hyperparameter $\Lambda$ of regularizer is tuned by training a non-private model on each dataset and set to be $10^{-2}$ for best performance. We use stochastic gradient descent algorithm to minimize Equation (\ref{eq-bg1}), of which the iterations and learning rate are set to be $100$ and $0.01$. The $\eps$ values we choose for training private models are between $0$ and $10$, which is very commonly used in related researches.

The Adult~\cite{Dua:2019}, Kddcup99~\cite{kddcup} and Gisette~\cite{guyon2004result} datasets are used in our experiments. 
The details of datasets is shown in Table~\ref{tb-dataset}. Due to the random noise addition, all experiments are repeated ten times and the average results are reported. The upper bound of $||\nabla \ell (\vtheta, \vz_i)||$ and eigenvalues of $\nabla^2 \ell(\vtheta, \vz_i)$ is set to be $2\sqrt{p}$ and $p$ due to~\cite{kifer2012private}, where $p$ denotes the dimensionality of dataset.

\begin{table}
\centering
	\begin{threeparttable}
	\centering
	\caption{Datasets used in our evaluations}
	\label{tb-dataset}
		\begin{tabular}{cccc}
		\toprule
		Dataset&\#Samples&\#Dim.&\#Classes\\
		\midrule
		Adult&$45,220$&$104$&$2$\\
		KDDCup99&$70,000$&$114$&$2$\\
		Gisette&$10,000$&$5,000$&$2$\\
		\bottomrule
		\end{tabular}
	\end{threeparttable}
\end{table}

\subsection{Experiment Design}
We totally set up three experiments. The first experiment is to verify the effectiveness of our approach. For each dataset, we randomly select $10,000$ samples and use them to train private models by the original objective perturbation mechanism according to each $\eps$ value, and use part of them as the measuring points to estimate empirical losses of other $\eps$ values according to Equation (\ref{eq-utility}), then compare estimation results with the real results. Due to our expectation, the gap between real results and calculated results should be very small. 

The second experiment is to analyze the impact of the selection of measuring points on the estimation results. We use each $\eps$ value as measuring point to estimate the empirical losses of all other $\eps$ values, and calculate the average error of estimated results. We hope that some rules can be found, which can be used as the guidance for selecting measuring point when using our approach.

The third experiment is to evaluate the effect of sample number on estimation error. Due to the limited sample number of Gisette, we only use Adult and Kddcup99 in this experiment. The number of samples is set from $1000$ to $20000$ for each dataset with fixed measuring points. According to our analysis in Section~\ref{sec-error-ana}, the estimation error should decrease with the increase of sample number.

\begin{figure*}
\centering
    \subfigure[Adult-LR]{
    \includegraphics[width=0.45\textwidth]{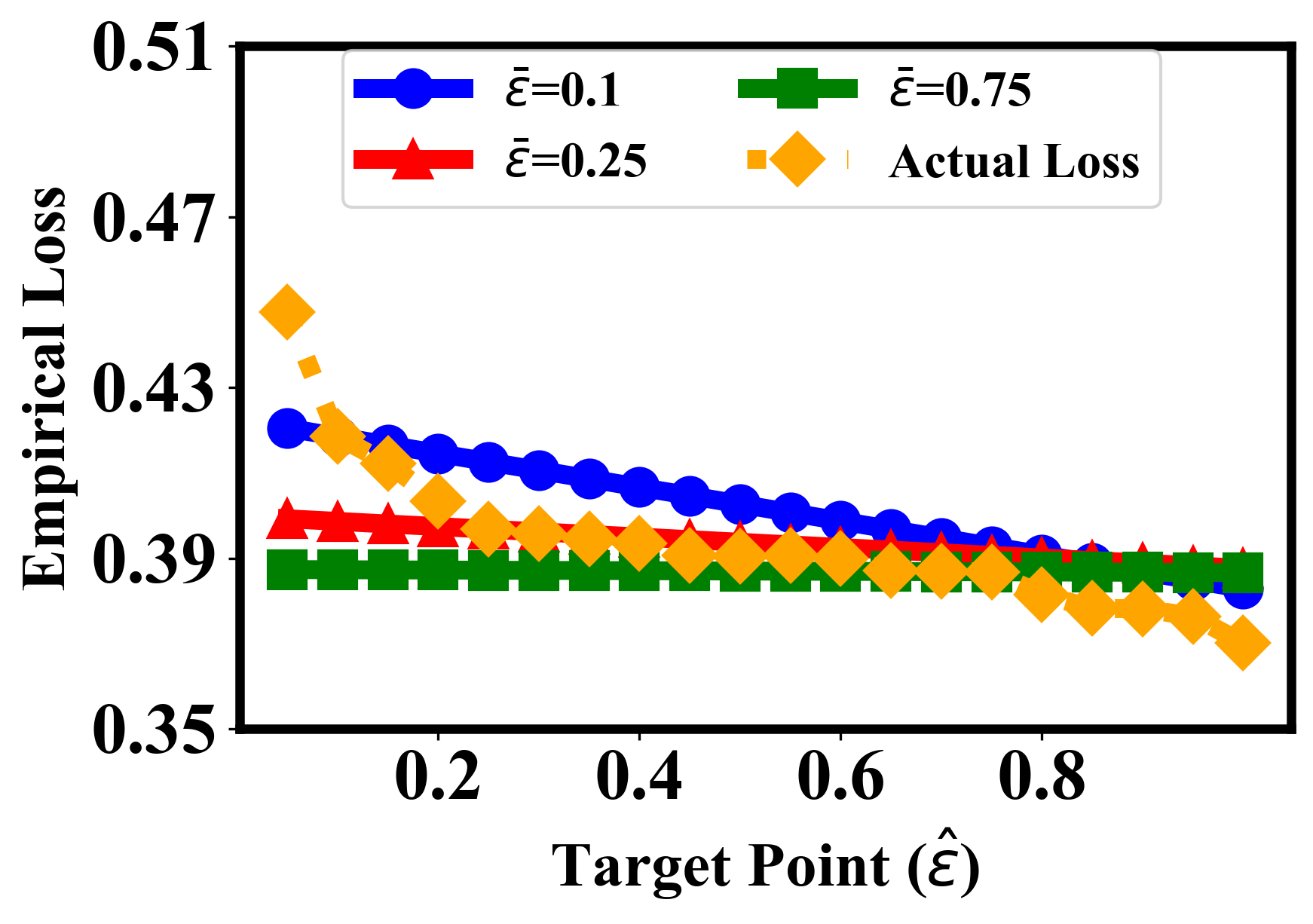}
    }
    \subfigure[Adult-SVM]{
    \includegraphics[width=0.45\textwidth]{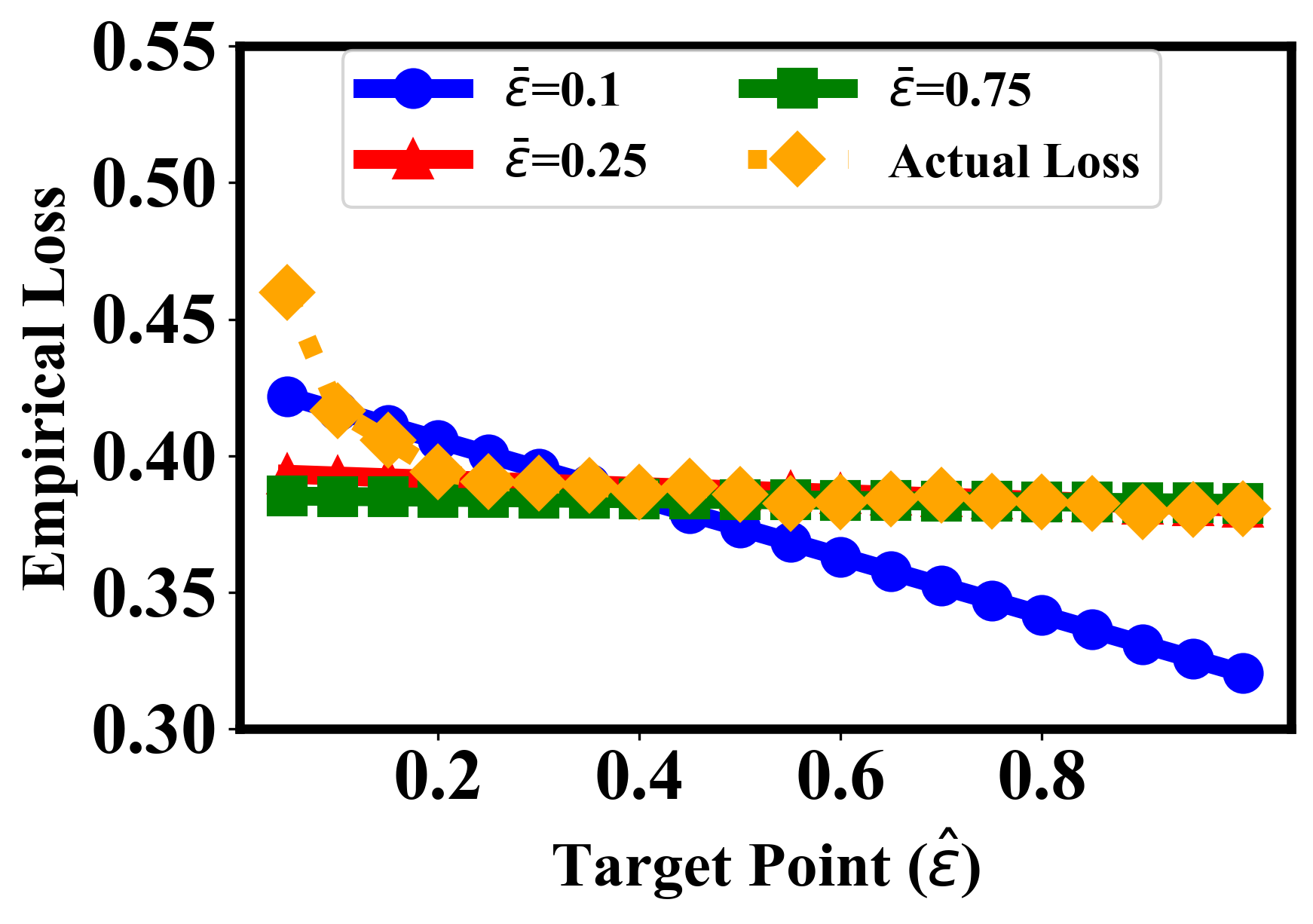}
    }
    \subfigure[Kddcup99-LR]{
    \includegraphics[width=0.45\textwidth]{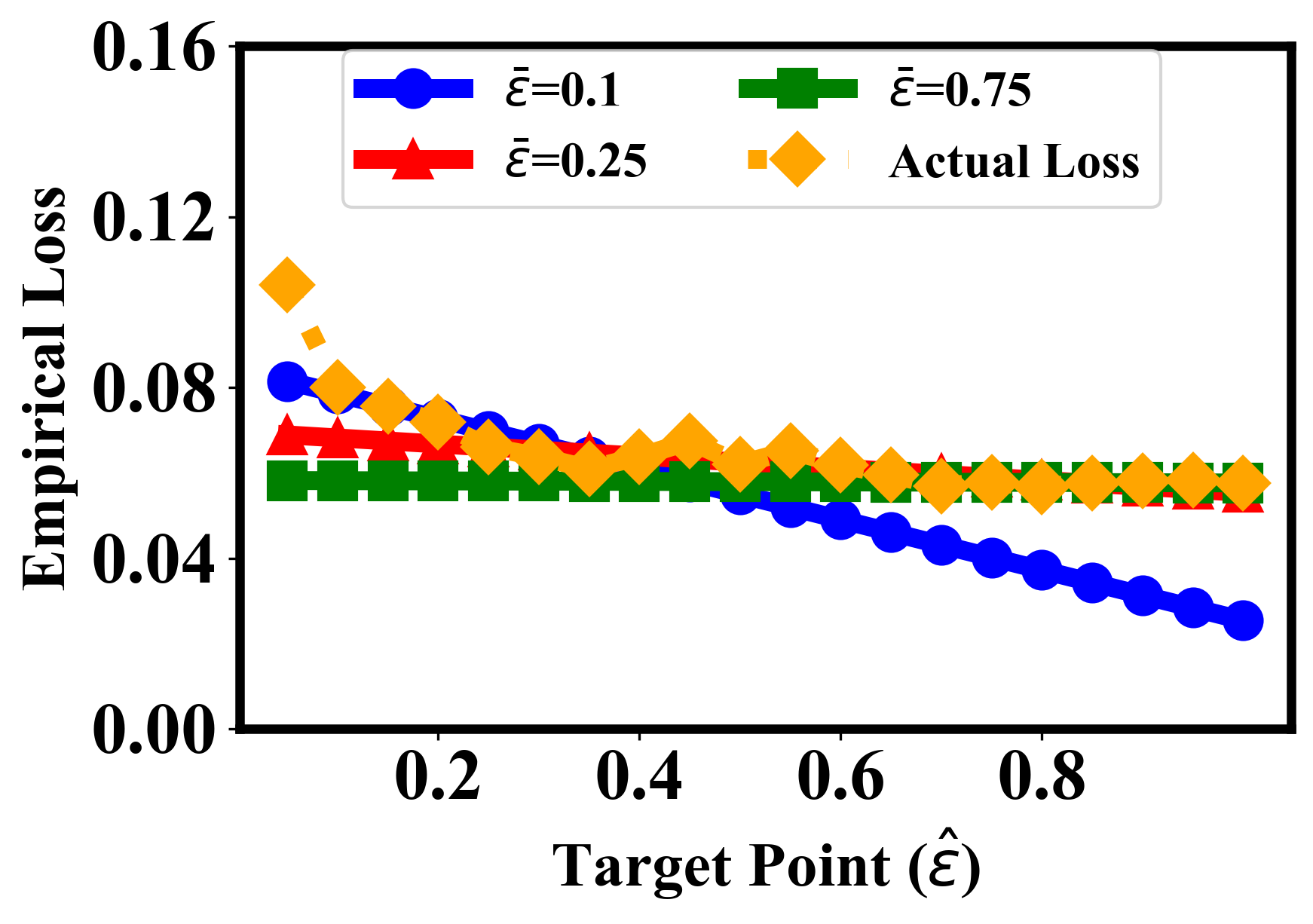}
    }
    \subfigure[Kddcup99-SVM]{
    \includegraphics[width=0.45\textwidth]{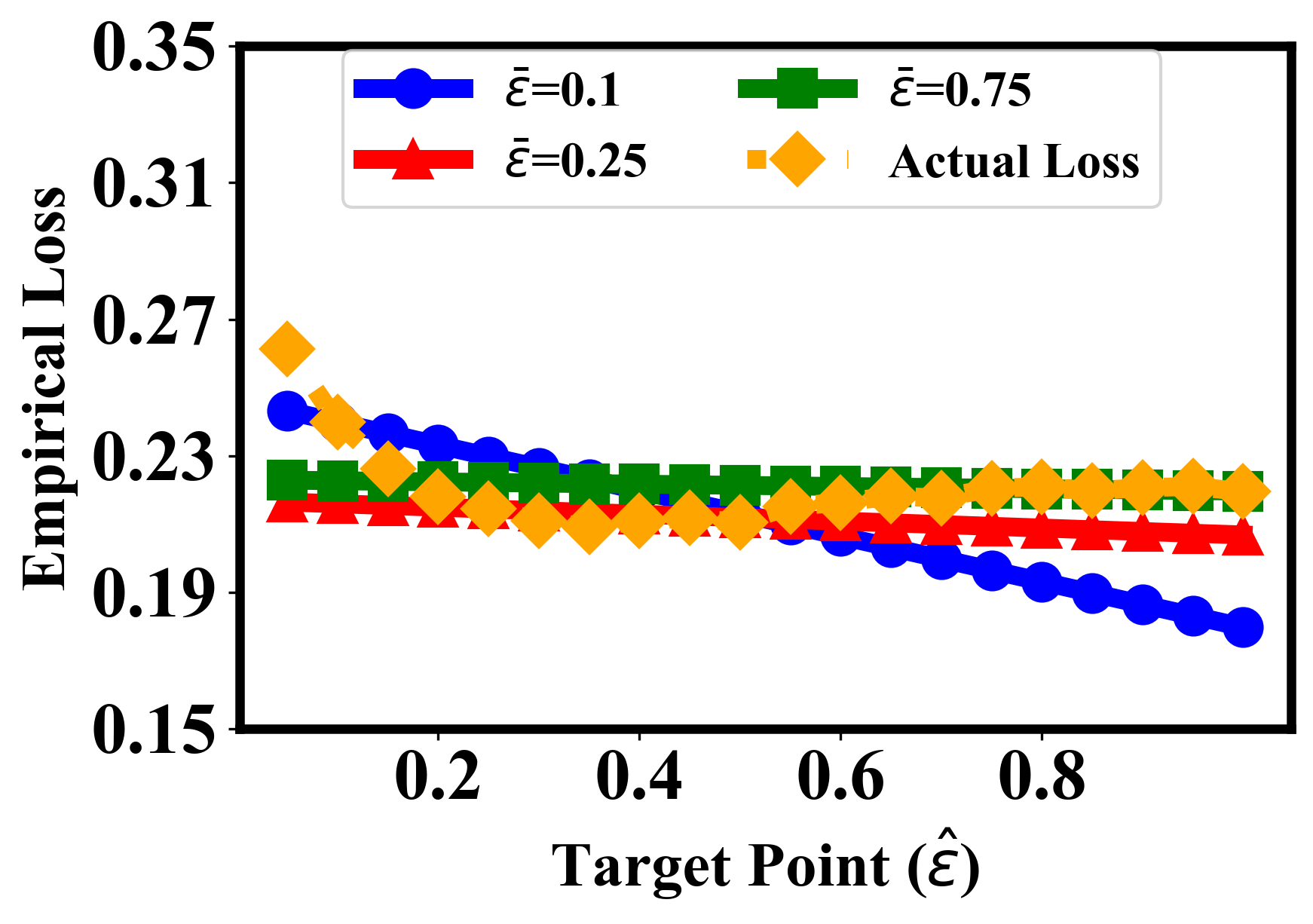}
    }
    \subfigure[Gisette-LR]{
    \includegraphics[width=0.45\textwidth]{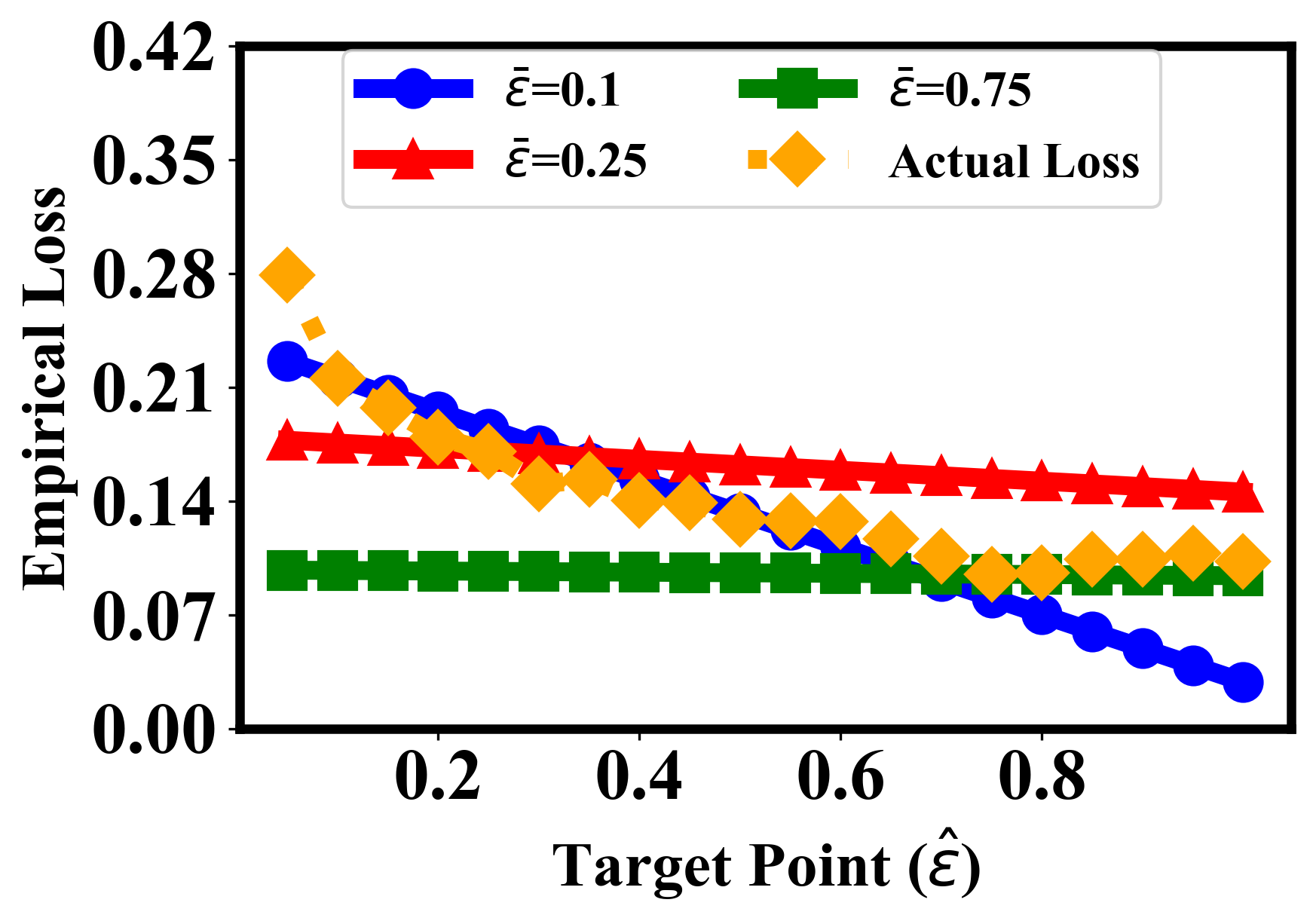}
    }
    \subfigure[Gisette-SVM]{
    \includegraphics[width=0.45\textwidth]{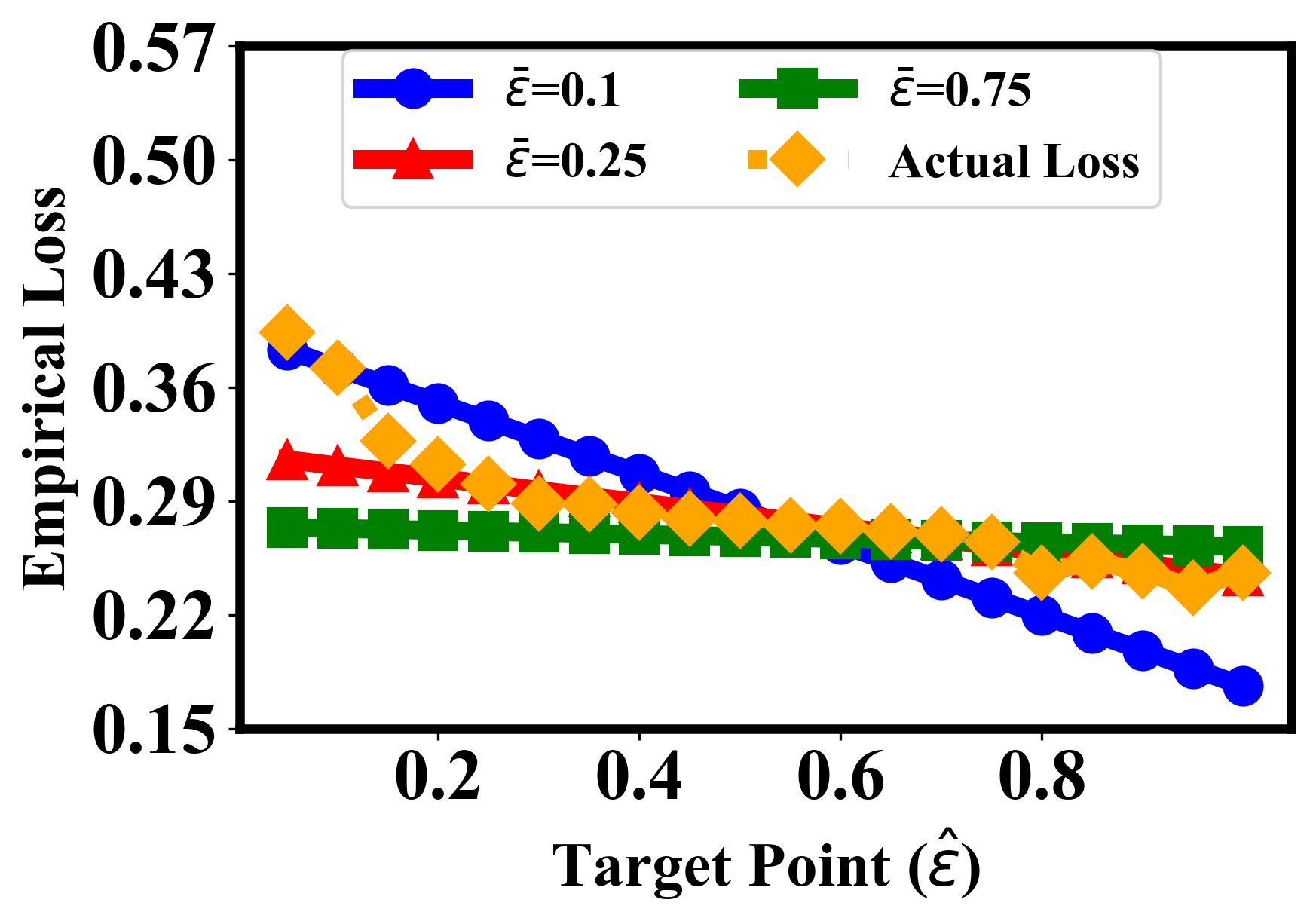}
    }
\caption{Performance of our approximation approach on datasets of Adult and Kddcup99 with logistic regression (LR) loss and Huber SVM (SVM) loss ($0 < \eps \le 1.0$).}
\label{fig-performance}
\end{figure*}

\begin{figure*}
\centering
    \subfigure[Adult-LR]{
    \includegraphics[width=0.45\textwidth]{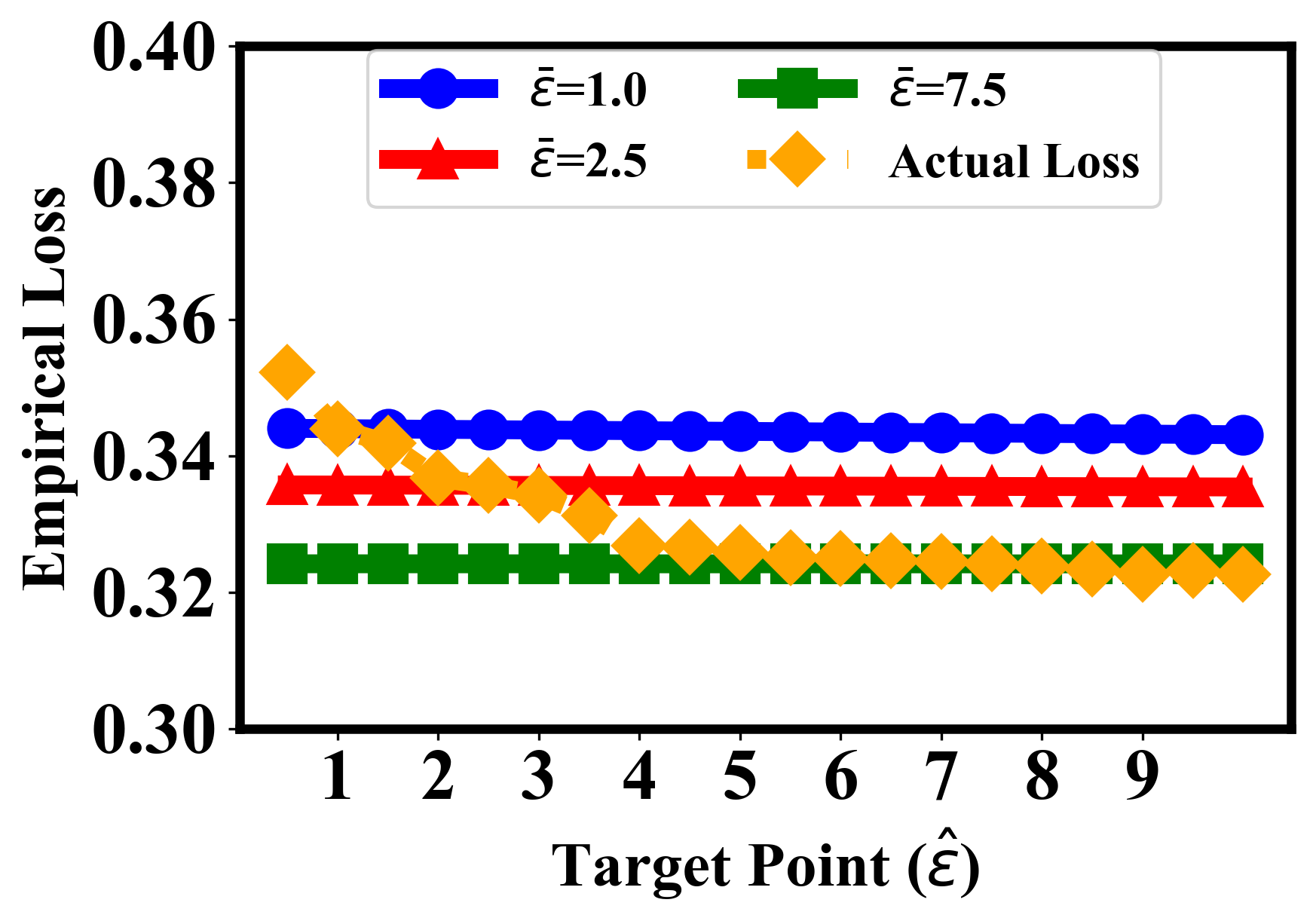}
    }
    \subfigure[Adult-SVM]{
    \includegraphics[width=0.45\textwidth]{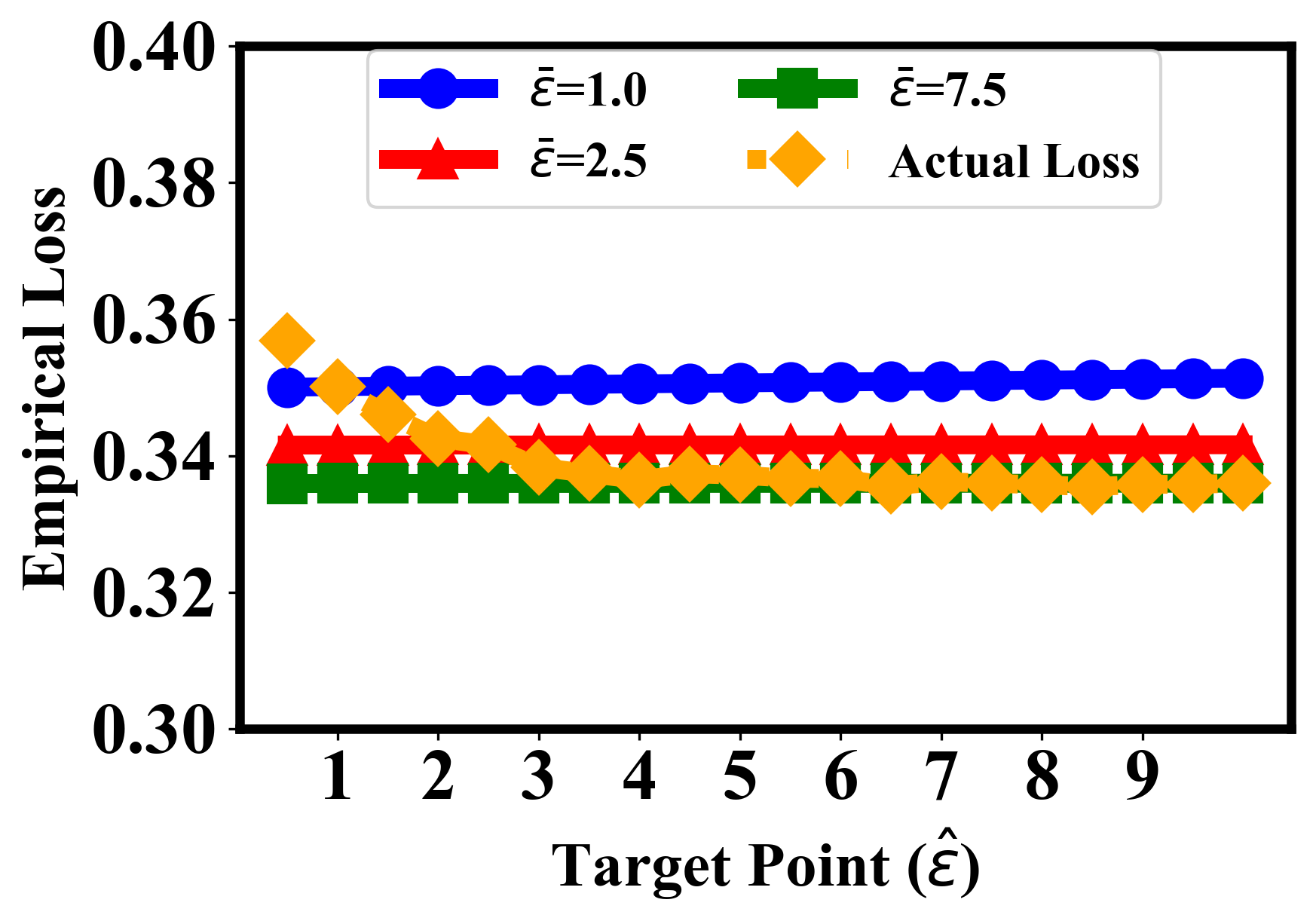}
    }
    \subfigure[Kddcup99-LR]{
    \includegraphics[width=0.45\textwidth]{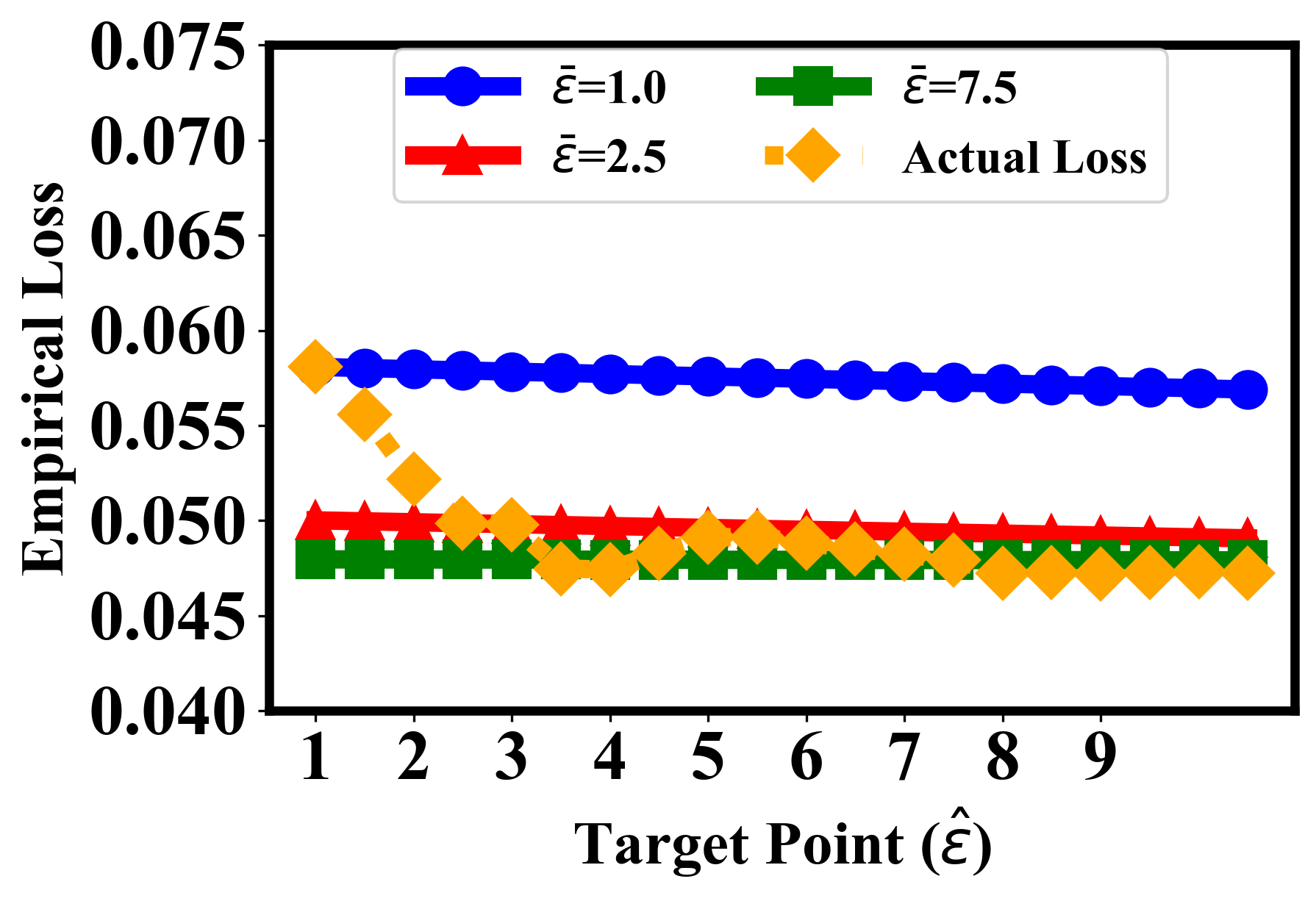}
    }
    \subfigure[Kddcup99-SVM]{
    \includegraphics[width=0.45\textwidth]{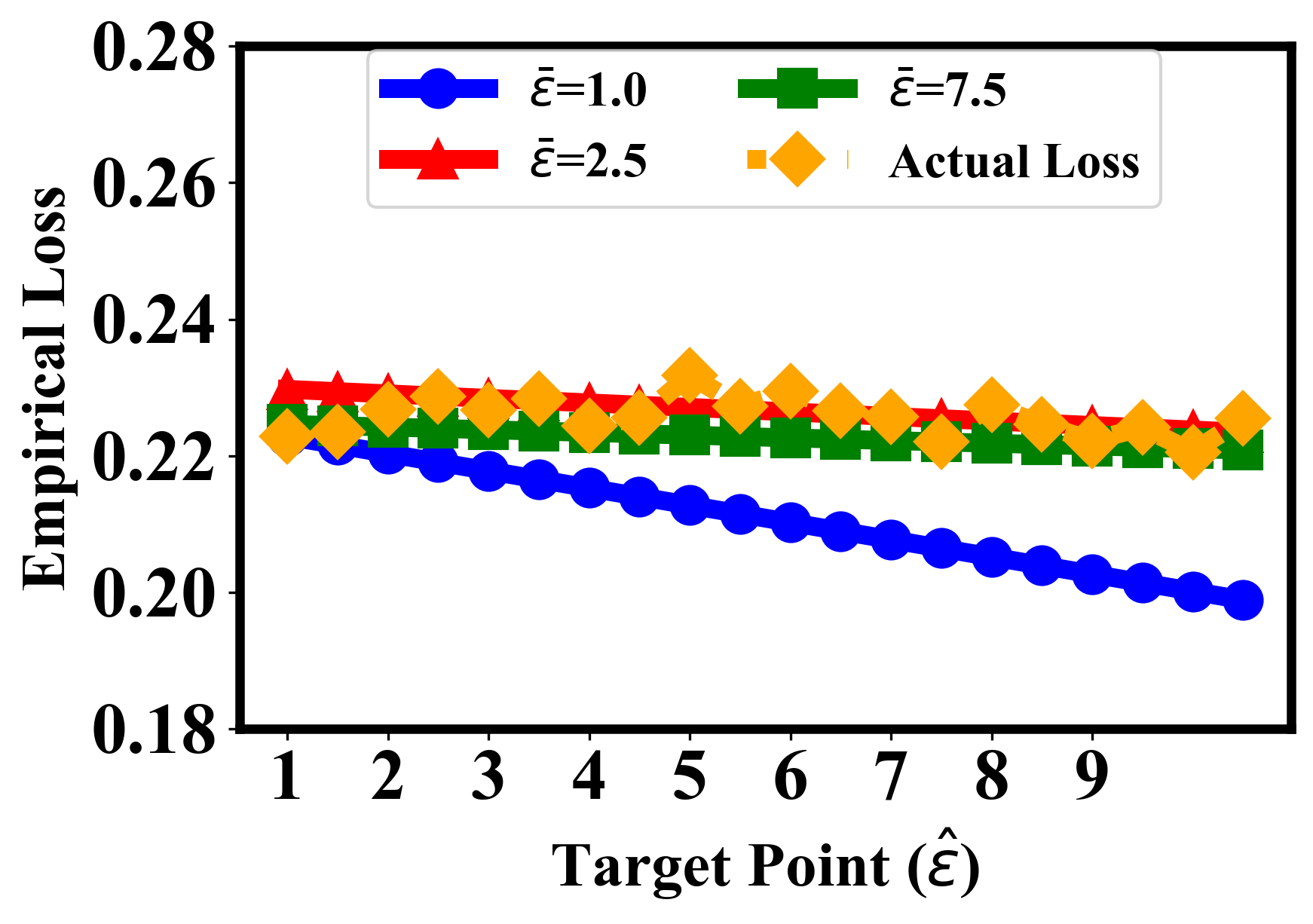}
    }
    \subfigure[Gisette-LR]{
    \includegraphics[width=0.45\textwidth]{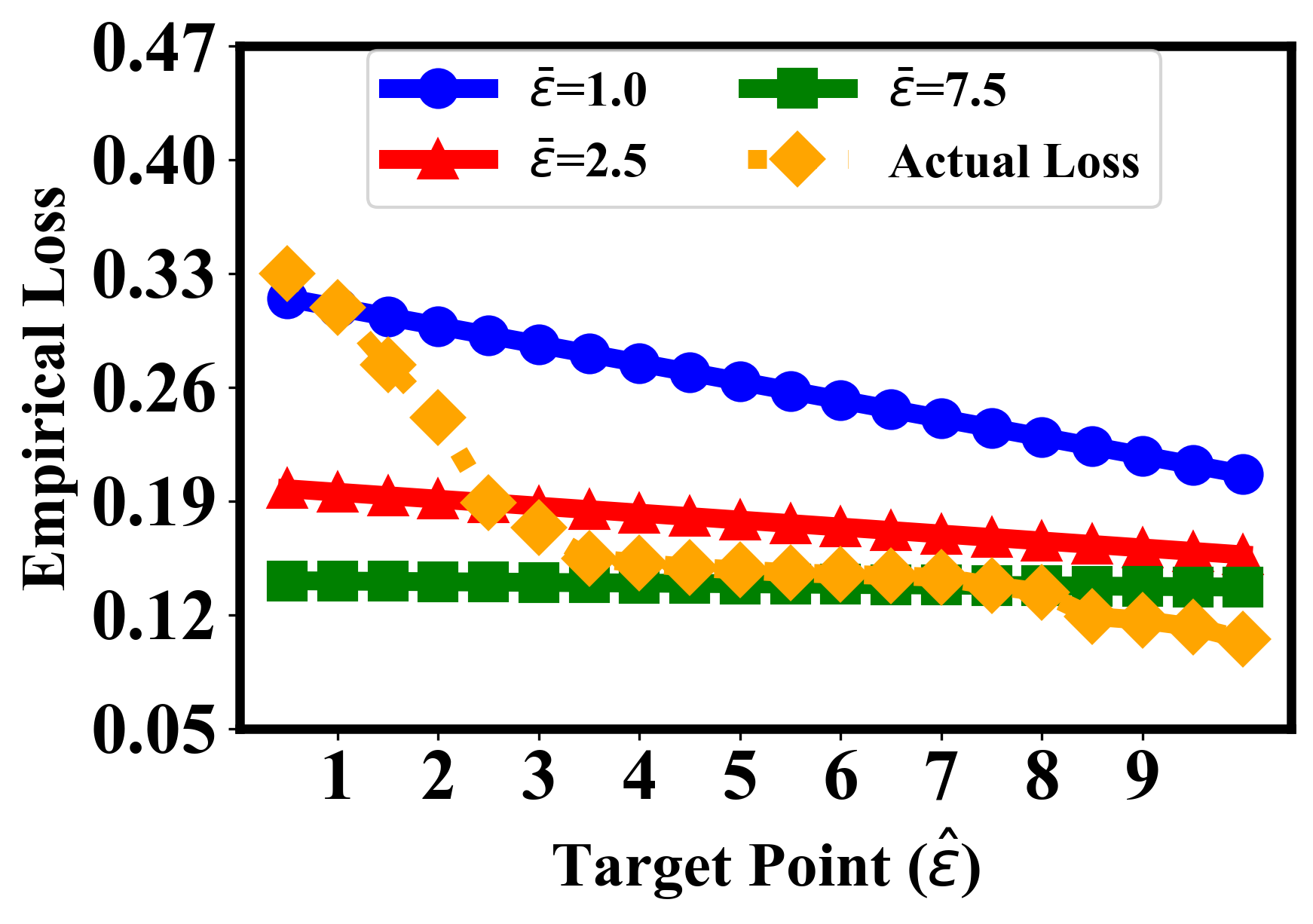}
    }
    \subfigure[Gisette-SVM]{
    \includegraphics[width=0.45\textwidth]{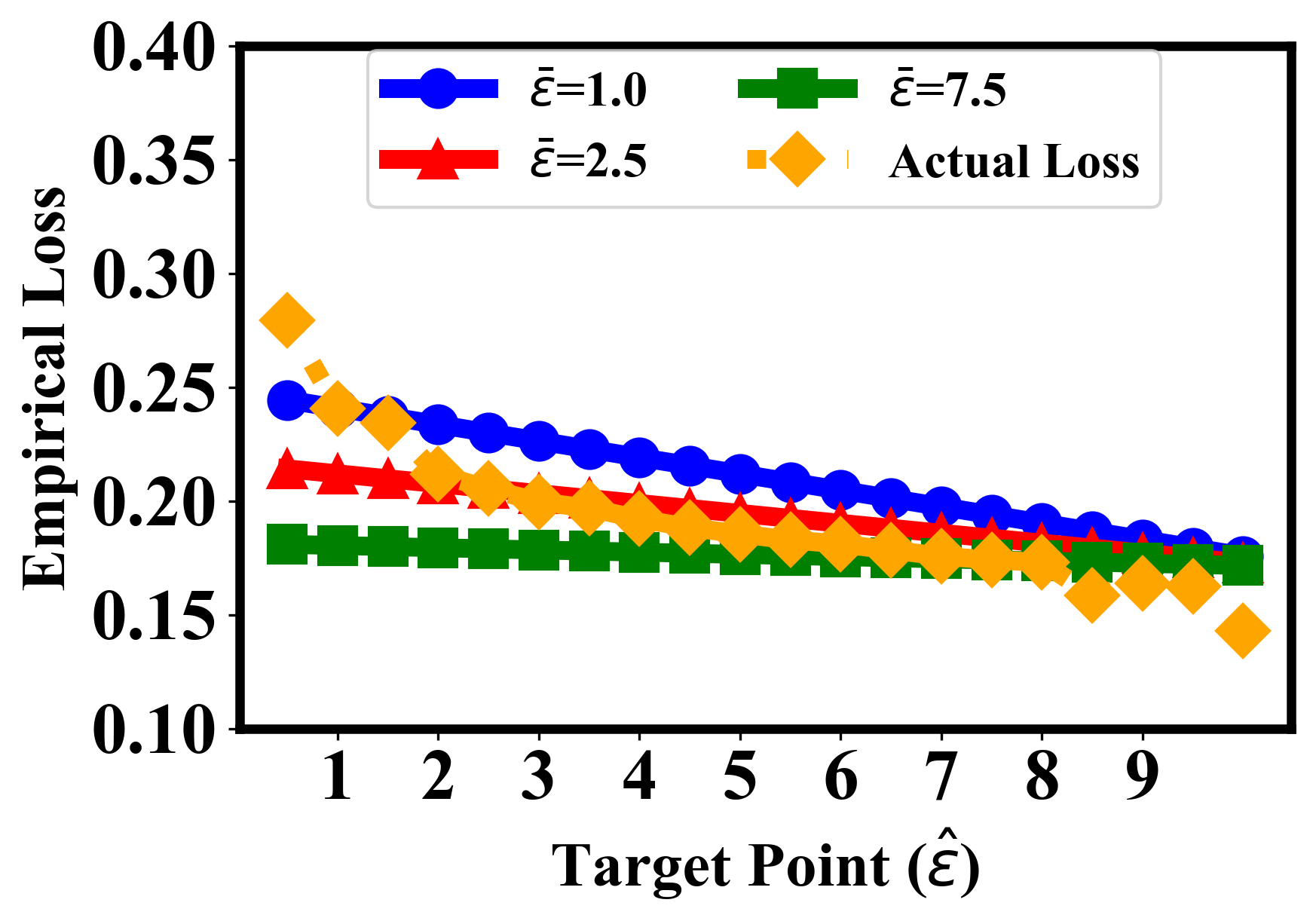}
    }
\caption{Performance of our approximation approach on datasets of Adult and Kddcup99 with logistic regression (LR) loss and Huber SVM (SVM) loss ($1.0 \le \eps \le 10$).}
\label{fig-performance-10}
\end{figure*}

\subsection{Evaluation Results and Analysis}
To make the results of the first experiment more clear, we divide $\eps$ into two domains due to the magnitude, which are $0 < \eps \le 1.0$ and $1.0 \le \eps \le 10$. For $0 < \eps \le 1.0$, we choose $\bar{\eps}=0.1, \bar{\eps}=0.25, \bar{\eps}=0.75$ as the measuring points to estimate the empirical loss at points varying from $\hat{\eps}=0.05$ to $\hat{\eps}=1.0$. For $1.0 \le \eps \le 10$, we choose $\bar{\eps}=1.0, \bar{\eps}=2.5, \bar{\eps}=7.5$ as the measuring points to estimate the empirical loss at points varying from $\hat{\eps}=1.0$ to $\hat{\eps}=10$. 

The experimental results of $0 < \eps \le 1.0$ are shown in Figure~\ref{fig-performance}, and the experimental results of $1.0 \le \eps \le 10$ are shown in Figure~\ref{fig-performance-10}. For both figures, the orange line representing the real empirical loss of each $\eps$ is called the real line, and the other lines representing the estimate results of empirical loss are called estimated lines. Each estimated line is drawn by the estimation results of the empirical loss of each $\hat{\eps}$ (target points). 
Obviously, from both Figure~\ref{fig-performance} and Figure~\ref{fig-performance-10}, it can be seen that the estimated lines basically fit the real line on each dataset (i.e. they are very close to each other). The results strongly confirm the effectiveness of our approach. Relatively speaking, the closer target point to measuring point, the higher estimation accuracy. This is basically consistent with our error analysis, that our error bound is proportional to the difference between $\bar{\eps}$ and $\hat{\eps}$. In addition, it can be seen that estimated lines are straight with different slope. The reason for this is that after training at the measuring point $\bar{\eps}$, the coefficients in (\ref{eq-thm1}) are determined, which makes a linear relationship between the value of target $\hat{\eps}$ and its estimated empirical loss. Compared with Figure~\ref{fig-performance} and Figure~\ref{fig-performance-10}, we can also see that when measuring and estimating on larger $\eps$, the error will be further reduced (note that the scale span of ordinate in Figure~\ref{fig-performance-10} is less than that in Figure~\ref{fig-performance}), which is in good agreement with the error analysis in Section~\ref{sec-error-ana}.

\begin{table*}[p]
\centering
\resizebox{0.99\linewidth}{!}{$
	\begin{threeparttable}
	\caption{Average error of estimated loss for measuring point from $\bar{\eps}=0.05$ to $\bar{\eps} = 0.50$ (part-1).}
		\begin{tabular}{ccccccccccc}
		\toprule
		$\bar{\eps}$&$0.05$&$0.10$&$0.15$&$0.20$&$0.25$&$0.30$&$0.35$&$0.40$&$0.45$&$0.50$\\
		\midrule
		Adlut-LR&$0.01982$&$0.00823$&$0.00908$&$0.00408$&$\bf0.00003$&$0.00033$&$0.00004$&$0.00037$&$0.00293$&$0.00309$ \\
		Adlut-SVM&$0.08188$&$0.02052$&$0.00327$&$0.00315$&$0.00445$&$0.00367$&$0.00385$&$0.00558$&$\bf0.00298$&$0.00585$\\
		Kddcup-LR&$0.05884$&$0.02393$&$0.00930$&$0.00721$&$0.00763$&$0.00845$&$0.00960$&$0.00677$&$\bf0.00283$&$0.00778$\\
		Kddcup-SVM&$0.04735$&$0.00866$&$0.00798$&$0.00793$&$0.00851$&$0.01052$&$0.01143$&$0.00954$&$0.00848$&$0.00967$\\
		Gisette-LR&$0.18427$&$0.01482$&$0.01181$&$0.02125$&$0.01931$&$\bf0.00302$&$0.00810$&$0.00372$&$0.00394$&$0.01361$ \\
		Gisette-SVM&$0.16245$&$0.00571$&$0.02480$&$0.00822$&$0.00567$&$0.01023$&$\bf0.00550$&$0.00712$&$0.00823$&$0.00823$ \\
		\bottomrule
\label{part-1}
		\end{tabular}
	\end{threeparttable}
	$}
\end{table*}

\begin{table*}[htb]
\centering
\resizebox{0.99\linewidth}{!}{$
	\begin{threeparttable}
	\caption{Average error of estimated loss for measuring point from $\bar{\eps}=0.55$ to $\bar{\eps} = 1$ (part-2).}
		\begin{tabular}{ccccccccccc}
		\toprule
		$\bar{\eps}$&$0.55$&$0.60$&$0.65$&$0.70$&$0.75$&$0.80$&$0.85$&$0.90$&$0.95$&$1.00$\\
		\midrule
		Adlut-LR&$0.00339$&$0.00623$&$0.00643$&$0.00649$&$0.01169$&$0.01480$&$0.00283$&$0.01485$&$0.01675$&$0.02276$ \\
		Adlut-SVM&$0.00833$&$0.00711$&$0.00546$&$0.00770$&$0.00780$&$0.00840$&$0.00890$&$0.01142$&$0.01053$&$0.01020$ \\
		Kddcup-LR&$0.01186$&$0.01288$&$0.01200$&$0.01149$&$0.01147$&$0.01213$&$0.00442$&$0.00768$&$0.01010$&$0.01274$ \\
		Kddcup-SVM&$0.00112$&$0.00137$&$\bf0.00038$&$0.00134$&$0.00192$&$0.00250$&$0.00486$&$0.00363$&$0.00154$&$0.00200$ \\
		Gisette-LR&$0.04637$&$0.04476$&$0.03693$&$0.03683$&$0.03308$&$0.03828$&$0.01420$&$0.01423$&$0.02461$&$0.03560$ \\
		Gisette-SVM&$0.01745$&$0.03641$&$0.02911$&$0.03548$&$0.04540$&$0.03603$&$0.00961$&$0.00926$&$0.01188$&$0.01327$ \\
		\bottomrule
\label{part-3}
		\end{tabular}
	\end{threeparttable}
	$}
\end{table*}


For the second experiment, we focus on the analysis of $0 < \eps \le 1.0$. Since there is quite a bit of experimental results, we divide the large horizontal table into three small horizontal tables (which are Table~\ref{part-1}, Table~\ref{part-3}) to ensure that the results are clearly displayed in the page. For each dataset (each line), the result in bold represents the measuring $\bar{\eps}$ with minimum error. We can see that the minimum error is concentrated in the middle, and the maximum error is either at the smallest $\bar{\eps}$ (i.e. $\bar{\eps} = 0.05$) or at the largest $\bar{\eps}$ (i.e. $\bar{\eps} = 1.00$). The reason for this is that the average gap between intermediate value and other values is smaller, so the average error of them tend to be smaller. In addition, the smallest $\bar{\eps}$ tends to have larger average error. This is because that $\tilde{\eps}$ in Equation (\ref{eq-bound}) would be small when measuring $\bar{\eps}$ is small, which would lead to high estimated error. Thus, in application, it would be better to select the measuring point near the mean value of target points so as to reduce the estimation error.

\begin{figure*}
\centering
    \subfigure[Adult-LR]{
    \includegraphics[width=0.45\textwidth]{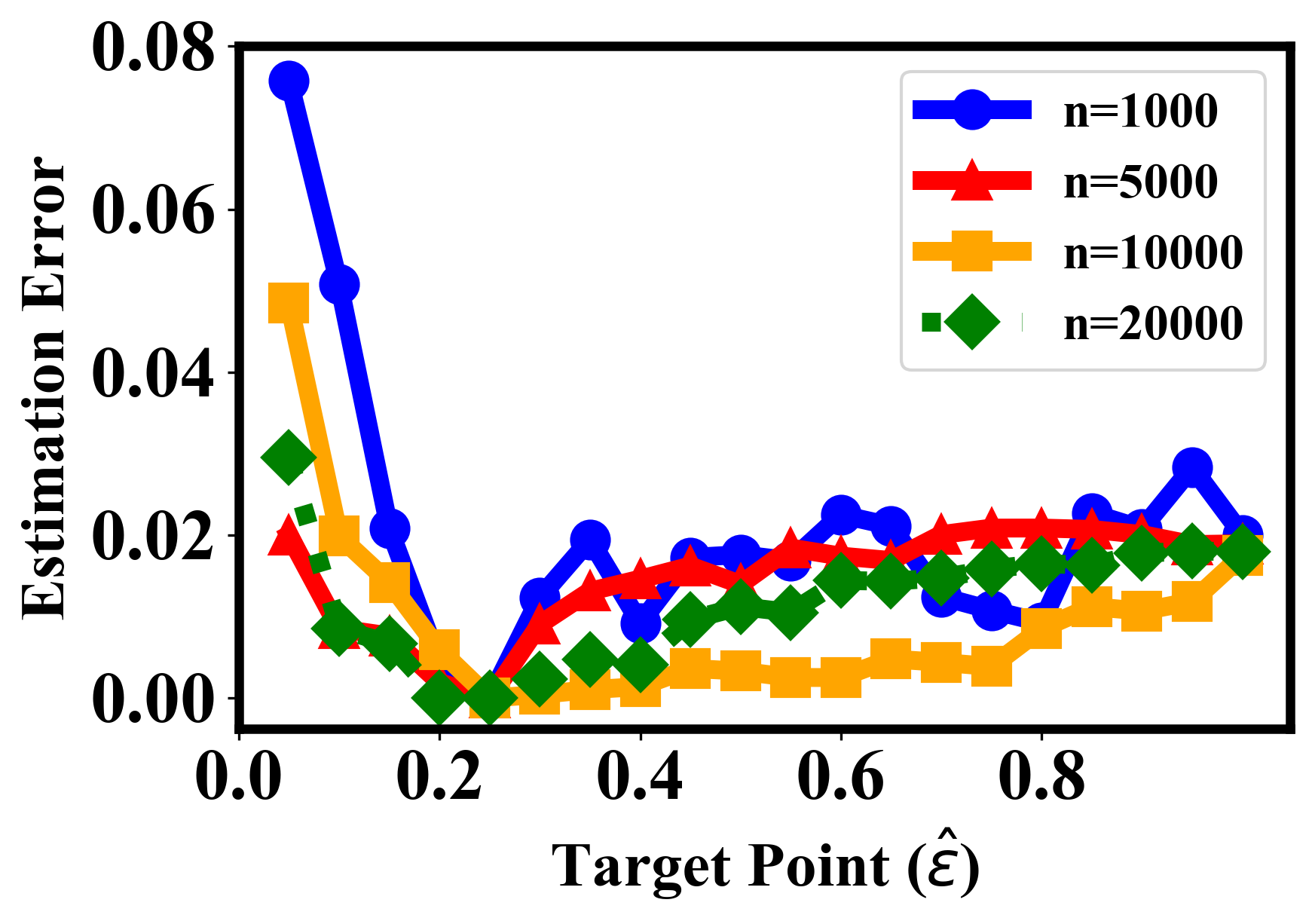}
    }
    \subfigure[Adult-SVM]{
    \includegraphics[width=0.45\textwidth]{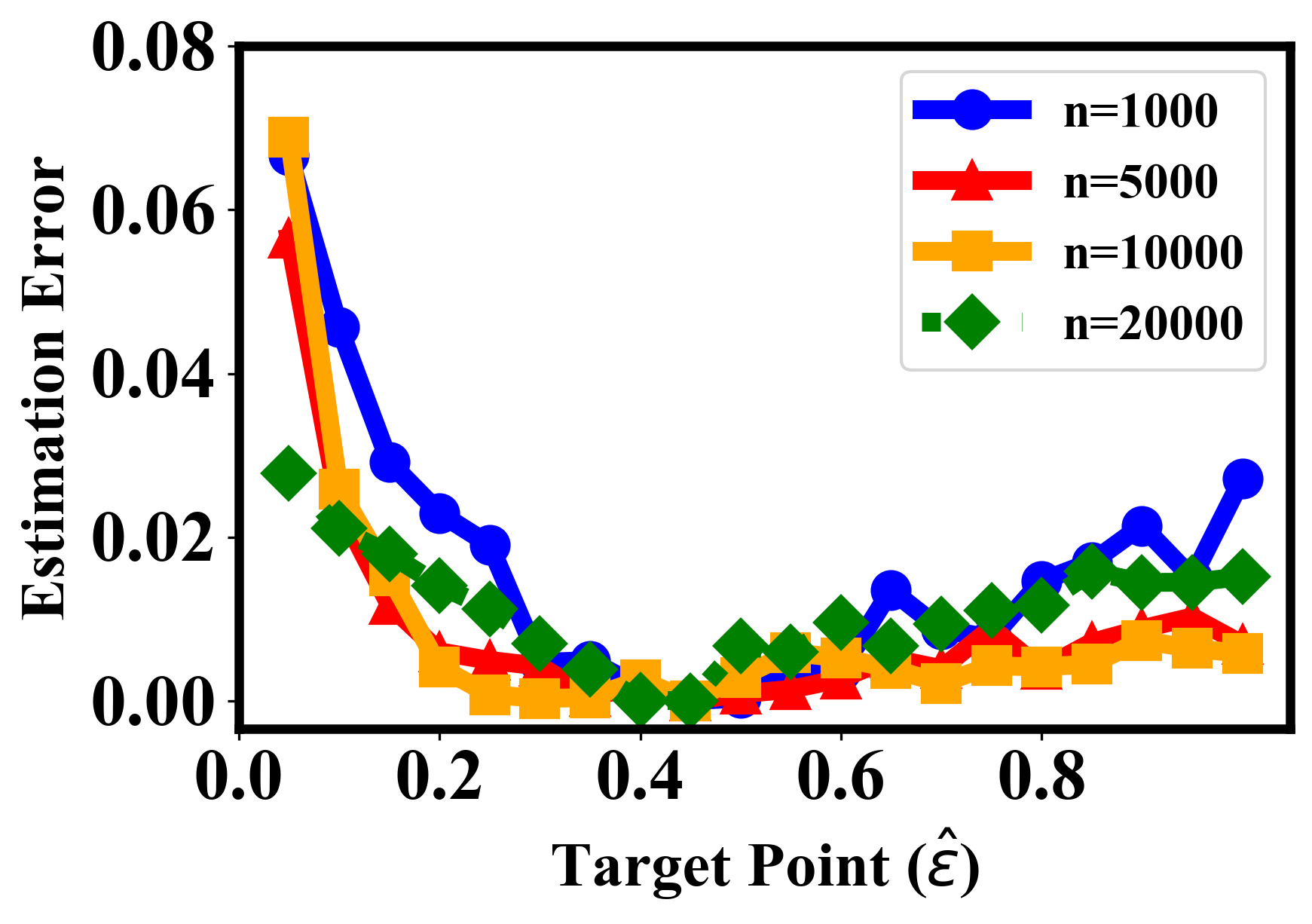}
    }
    \subfigure[Kddcup99-LR]{
    \includegraphics[width=0.45\textwidth]{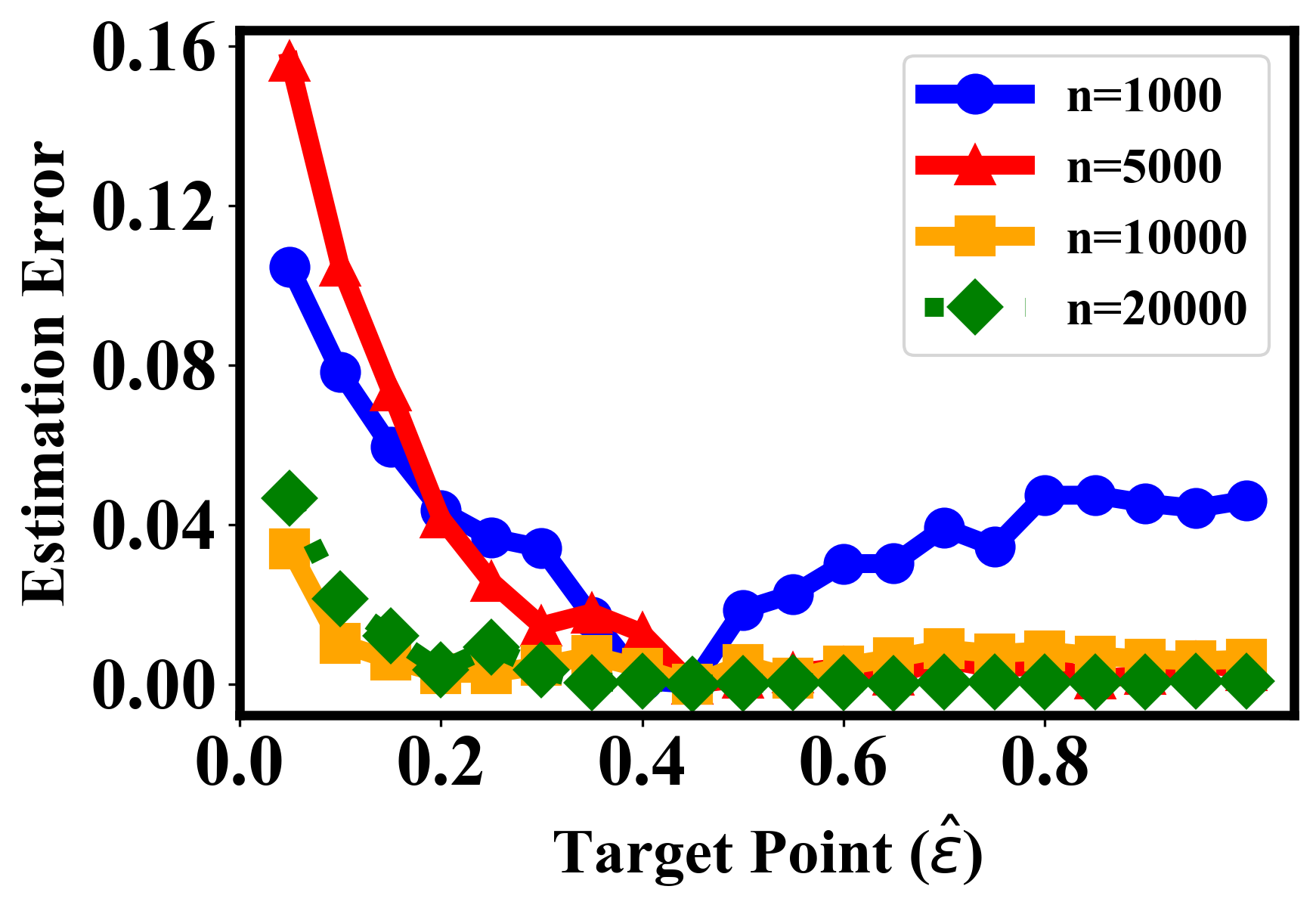}
    }
    \subfigure[Kddcup99-SVM]{
    \includegraphics[width=0.45\textwidth]{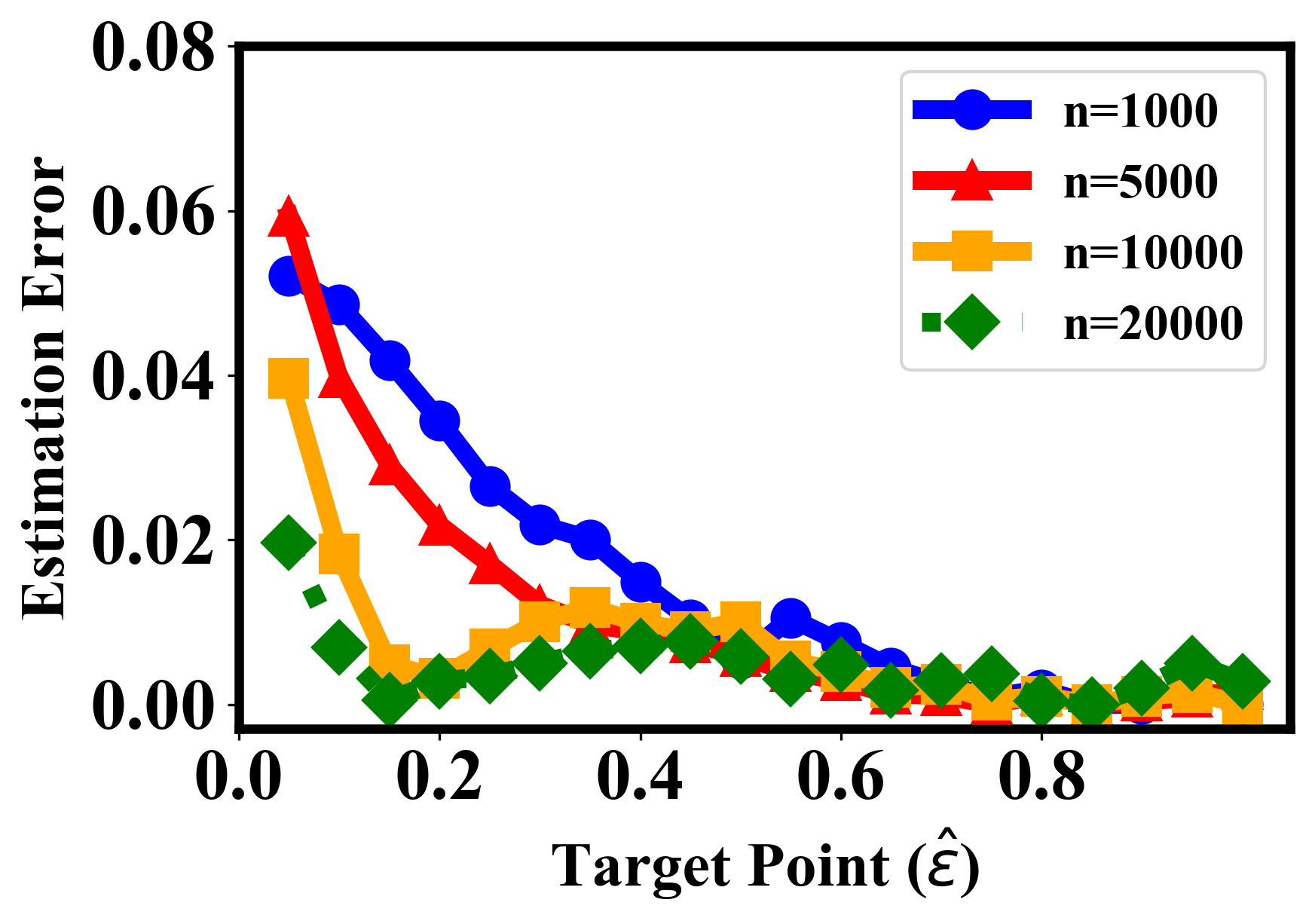}
    }
\caption{Affect of sample number on estimation error. The measuring point for each evaluation, from left to right, is set to be $\bar{\eps}=0.25, \bar{\eps}=0.45, \bar{\eps}=0.45, \bar{\eps}=0.85$, respectively, which achieves the minimum error in the second experiment.}
\label{fig-number}
\end{figure*}

For the third experiment, we fix the value of $\bar{\eps}$ for each dataset to be the one that achieves minimum error and vary the sample number  from $1000$ to $20000$ in the second experiment. 
The results are shown in Figure~\ref{fig-number}. Since the objects we compare are all estimation results, we directly use estimation error (instead of empirical loss we used in the first experiment) as the ordinate of each figure. Obviously, the estimation error decreases with the increase of sample number. This is consistent with our analysis in Section~\ref{sec-error-ana} as the error bound of our approach is inversely proportional to sample number $n$ due to Equation~(\ref{eq-bound}). One possible explanation for this is that more training samples can improve the accuracy of the trained model and reduce the empirical loss at measuring $\bar{\eps}$. From Figure~\ref{fig-number}, we can see our error analysis is confirmed again, that the closer to measuring $\bar{\eps}$, the lower the estimation error of target $\hat{\eps}$, and under the same distance, the estimation error of larger target $\hat{\eps}$ is less than that of smaller target $\hat{\eps}$.

\section{Related Works}

To overcome the conservatism of ``utility theorems'',~\cite{Ligett2017accuracy} proposed an empirical approximation method for utility analysis. In their method, they formalize the utility analysis as finding the minimum $\eps$ that meets a sufficient level of accuracy. They start from a very private $\eps$, then gradually subtract the noise until the expected accuracy is achieved.
The advantage of this method is that it focuses on the reality, not the theoretical worst-case. Therefore, their analysis is much closer to the real value than ``utility theorems''. However, an obvious drawback of this method is that it requires so many attempts of $\eps$ (noise subtraction), that would take a lot of computation and training time. This is often difficult to be accepted in practice, especially in large-scale learning tasks. Compared with~\cite{Ligett2017accuracy}, our approach can effectively overcome the shortcomings of the ``utility theorems'' and the method in~\cite{Ligett2017accuracy}. On the one hand, our approach is a data-dependent theory, thus the analysis tends to be consistent with the actual situation. On the other hand, instead of repeated training to approach the target $\eps$ step by step, our approach only needs one-round training to obtain the target $\eps$ directly, which makes it much more efficient than~\cite{Ligett2017accuracy}.

The accuracy-first differential privacy has also been widely studied in differentially private query answering systems. For example, GUPT~\cite{mohan2012gupt} offers a tool based on the sample-and-aggregate framework for differential privacy, which allows analysts to specify the target accuracy of the output, and compute privacy from it—or vice versa. PSI~\cite{marco2016psi} offers to the data analyst an interface for selecting either the level of accuracy that she wants to reach, or the level of privacy she wants to impose. APEx~\cite{ge2019apex} allows data analysts to pose adaptively chosen sequences of queries along with required accuracy bounds. DPella~\cite{lobo2020a} provides a programming framework for reasoning about privacy, accuracy, and their trade-offs. It uses taint analysis to detect probabilistic independence and derive tighter accuracy bounds using Chernoff bounds. Different from the above studies, our approach focuses on the empirical risk minimization problems and uses empirical error of the learning algorithms to pick the epsilon. Besides, our approach can be directly applied to any objective perturbed learning algorithms.

\section{Conclusion}
In this paper, we focus on the practical issue of differentially private machine learning: how to assess the impact of different $\eps$ values on utility of the well-trained model in advance. 
We formalize this problem as the following question: how $\eps$ exactly affects the utility. 
We find that most existing researches of differentially private machine learning refer to this issue. Usually, they define this analysis as ``utility analysis'' and conclusively give the utility bound based on ``utility theorem''. However, ``utility theorems'' often fail to provide meaningful assessments in practice as they tend to be worst-case bounds and would be too conservative as a practical reference. 

Different from previous researches, we address this issue through ``empirical'' theorem. By tracing the change of utility with modification of $\eps$ in objective perturbation mechanisms, we find that there is an established relationship between $\eps$ and utility. We reveal this relationship through a formulaic definition and then put forward a practical approach for estimating the utility of the trained model under an arbitrary value of $\eps$. 
Both theoretical analysis and experimental results on real-world datasets demonstrate the good estimation accuracy and broad applicability of our approximation approach. To the best of our knowledge, it is the first time to formalize such a relationship between $\eps$ and utility under practical cases. The new understanding of utility-privacy trade-off in differential privacy provided by our analysis opens a huge space for exploration in theoretical research and practical application. As providing algorithms with strong utility guarantees that also give privacy when possible becomes more and more accepted, we believe that our approach would have high practical value and can provide meaningful guidance for addressing practical issues.

In our future work, we wish to further improve the generality of our approximation approach, such as extend it to more general perturbation mechanisms (i.e. gradient perturbation) or give further analysis of our approximate approach for general cases (i.e. non-convex, non-convergent and non-differentiable cases).

\bibliographystyle{unsrtnat}
\bibliography{ms}  






\end{document}